\documentclass[aps, pra, a4paper, longbibliography, twocolumn, superscriptaddress]{revtex4-1}

\usepackage[utf8]{inputenc}

\usepackage{amsmath,amssymb,amsthm}
\usepackage{mathtools}
\usepackage{braket}
\usepackage{comment}
\usepackage{graphics}
\usepackage{hyperref}
\usepackage{xcolor}
\usepackage{dsfont}
\hypersetup{citecolor=blue}
\hypersetup{colorlinks=true}
\hypersetup{linkcolor=blue}
\hypersetup{urlcolor=blue}

\theoremstyle{definition}
\newtheorem{theorem}{Theorem}
\newtheorem{lemma}[theorem]{Lemma}
\newtheorem{corollary}[theorem]{Corollary}

\renewcommand{\braket}[2]{\langle #1\vert #2\rangle}
\definecolor{darkred}{rgb}{.8 .1 .1}
\definecolor{darkgreen}{rgb}{.1 .8 .1}
\definecolor{darkyellow}{rgb}{.6 .6 .0}

\newcommand{\GHZ}{\text{GHZ}}

\newcommand{\ketbra}[1]{| #1\rangle \langle #1|}

\newcommand{\be}{\begin{equation}}
\newcommand{\ee}{\end{equation}}
\newcommand{\bea}{\begin{eqnarray}}
\newcommand{\eea}{\end{eqnarray}}

\newcommand{\SU}{\text{SU}}
\newcommand{\kommentar}[1]{}

\newcommand{\forget}[1]{}

\newcommand{\rk}{\text{rk}}

\newtheorem{observation}{Observation}

\begin{document}
%%%%%%%%%%%%%%%%%%%%%%%%%%%%%%%%%%%%%%%%%%%%%%%%%%%%%%%%%

\title{
Finding maximal quantum resources
} 
%Maximizing the geometric measure of entanglement
\author{Jonathan Steinberg}
\email{steinberg@physik.uni-siegen.de}
\affiliation{Naturwissenschaftlich-Technische Fakultät, 
Universität Siegen, Walter-Flex-Straße 3, 57068 Siegen, Germany}
\affiliation{State Key Laboratory for Mesoscopic Physics, School of Physics and Frontiers Science Center for Nano-Optoelectronics, Peking University, Beijing 100871, China}

\author{Otfried Gühne}
\email{otfried.guehne@uni-siegen.de}
\affiliation{Naturwissenschaftlich-Technische Fakultät, 
Universität Siegen, Walter-Flex-Straße 3, 57068 Siegen, Germany}

\date{\today}

\begin{abstract}
For many applications the presence of a quantum advantage  crucially depends on the availability of resourceful states. 
Although the resource typically depends on the particular task, in the context of multipartite systems entangled quantum states are often regarded as resourceful. We propose an algorithmic method to find  
maximally resourceful states of several particles for various applications and quantifiers. We discuss in detail the case of the geometric measure, identifying physically interesting states and also deliver insights to the problem of 
absolutely maximally entangled states. Moreover, we demonstrate the universality of our approach by applying it to maximally entangled subspaces, the Schmidt-rank, the stabilizer rank as well as the preparability in triangle networks. 
\end{abstract}
\maketitle

%%%%%%%%%%%%%%%%%%%%%%%%%%%%%%%%%%%%%%%%%%%%%%%%%
%%%%%%%%%%%%%%%%%%%%%%%%%
%%%%%%%%%%%%%%%%%%%%%%%%%%%%%%%%%%%%%%%%%%%%%%%%%

\section{Introduction}
The access to multipartite quantum states is an indispensable prerequisite for many applications in quantum information, turning them into a powerful resource which potentially outperform their classical counterparts~\cite{BB84, Ekert91, Metrology_outperform_1993}. Indeed, magic states turn out to be a resource for fault-tolerant quantum computation~\cite{logic_gate_construction_chuang_2000,universal_with_cliffords_bravyi_2005} while cluster states are resourceful for measurement-based quantum computation~\cite{one_way_qc_raussendorf_2001,measurement_based_cluster_briegel_2003}. Furthermore, the power of quantum metrology heavily relies on the ability to prepare multipartite quantum states. However, for a particular given task it is in general very challenging to identify those multipartite states which yield the largest advantage.

For many important applications entanglement has been proven to be a powerful resource.  An example of resourceful states are the absolutely maximally entangled (AME) states which maximized the entanglement in the bipartitions, but are notoriously difficult to characterize~\cite{scott2004, max_multi_entangles_2008, exploring_entanglement_2018, gour2010, huber2017, bounds_shadow_huber_2018,utzekatze_2022}. 
Still, the analysis of AME states is important for understanding quantum error correction and regarded as one of the central problems in the field~\cite{5open_question_horodecki_2022,rather2021}.
However, multiparticle entanglement offers a complex and rich structure resulting in the impossibility of quantification by means of a single number. 
Consequently, there is  a variety  of quantifiers, each emphasizing a different property that makes a state a valuable resource~\cite{vedral1998, vidal1998, giraud2010}.

\begin{figure}[t]
    \centering
   \includegraphics[width=0.85\columnwidth]{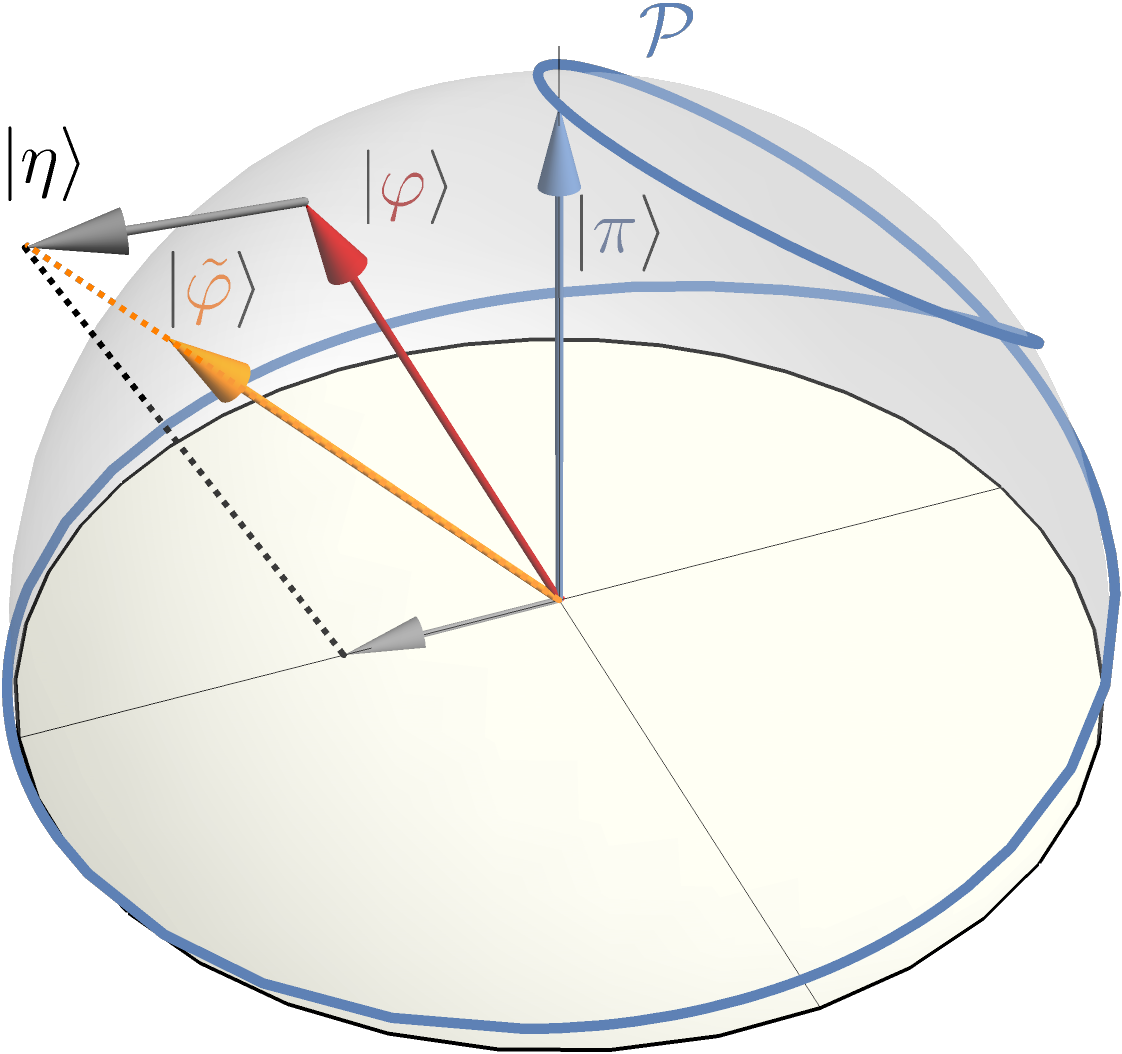}
    \caption{Schematic illustration of the iteration step of the algorithm. The set of all states is represented by the half sphere and the set of product states by the lower dimensional manifold $\mathcal{P}$. If the algorithm is initialized in state $\vert \varphi \rangle$ (red arrow), we first compute the best approximation within $\mathcal{P}$, denoted by $\vert \pi \rangle$ (blue arrow). Then, we compute the projector into the orthocomplement of $\vert \pi \rangle$, which is here given by the $xy$-plane. The portion of $\vert \varphi \rangle$ within the $xy$-plane is given by $\vert \eta \rangle$ (gray arrow). The new state $\vert \tilde{\varphi}\rangle \sim \vert \varphi \rangle + \epsilon \vert \eta \rangle$ is then the normalized version of $\vert \varphi \rangle$ shifted by a small amount $\epsilon>0$ into the direction $\vert \eta \rangle$.}
    \label{fig:convergence}
\end{figure}

The {geometric measure} of
entanglement~\cite{geo_measure_shimony_1995,monotones_barnum_2001,Wei2003,brukner_zukowski_2008}
quantifies the proximity of a quantum state to the set 
of product states, has an intuitive meaning and also offers 
multiple operational interpretations. For instance, it relates to multipartite 
state discrimination using LOCC~\cite{Hayashi2006}, the additivity of channel
capacities~\cite{channel_capacity_2010}, quantum state
estimation~\cite{state_estimation_wei_2011} and was also used to describe quantum phase transitions~\cite{orus_phase_transition_2010,
topology_geometric_measure_2014, lattice_phase_transition_2016, 
cluster_xy_model_2019}. 
%Further, it has been realized that generic 
%quantum states are highly entangled~\cite{Hayden2009}. 
In complexity theory, 
identifying maximally entangled states and computing their geometric measure
allows for the identification of cases where the $\text{MAX-}N$-local 
Hamiltonian problem and its product state approximation deviate
maximally~\cite{gharibian2012, tensor_norms_information_montanaro}. 
So, although high  geometric entanglement does not guarantee that a quantum state 
is useful for all tasks \cite{one_way_qcomputing_2007,Bremner2009,gross2009},
finding maximally entangled states has been recognised as a natural 
and important problem~\cite{tensor_norms_information_montanaro}.
So far, however, maximally entangled states have only been identified 
within the low-dimensional family of symmetric qubit states, where their 
computation is related to the problem of distributing charges on 
the unit sphere~\cite{aulbach2010,giraud2010,symmetric_maximizer_braun_2010}, or within the 
family of graph states that stem from bipartite 
graphs~\cite{direct_evaluation_graph_states_2013}.

Mathematically, the complexity of the task reflects the fact that 
pure multiparticle states are described by tensors. In contrast 
to the matrix case, notions like ranks and eigenvalues are for 
tensors much less understood and their computation turns out 
to be a hard problem~\cite{entanglement_tensor_rank_2008, 
tensor_problems_hard_2013}. Interestingly, the geometric measure 
is closely related to the recently introduced concept of tensor 
eigenvalues~\cite{Wei2003, tensor_eigenvalues_qi_2005, Chen2010, 
measure_tensor_u_eigenvalues_2014, Hu2016}, offering a much more complex structure as the matrix case~\cite{eigenstructure_raasch_2023} as well as to the 
notion of injective tensor norms~\cite{tensor_norm_additivity_2002, aubrun2017} 
and matrix permanents~\cite{wei2010}. Here, maximally entangled states offer 
maximal tensor eigenvalues~\cite{qi2018} and it was conjectured that the 
overlap of a multipartite qubit state with the set of product states 
decreases exponentially in the number of particles~\cite{aubrun2017}. 
So, the identification of maximally entangled states provides 
valuable intuition to decide this conjecture.

In this paper we design an iterative method for finding maximally resourceful multipartite quantum states. Choosing initially 
a generic quantum state, we show that in each step of the algorithm the resourcefulness increases. We illustrate the universality of our method by applying it to various different resource quantifiers and present a detailed analysis for the geometric measure. Here we 
identify for moderate sizes the 
corresponding states, revealing an interesting connection to AME states. 
Further, we introduce a novel quantifier for maximally entangled subspaces for which we provide a full 
characterization for the case of three qubits.

%%%%%%%%%%%%%%%%%%%%%%%%%%%%%%%%%%%%%%%%%%%%%%%%%%

\section{The geometric measure}
This measure quantifies how well a given multiparticle quantum state 
can be approximated by pure product states. More formally~\cite{geo_measure_shimony_1995, monotones_barnum_2001, Wei2003, brukner_zukowski_2008}, given 
$\vert \varphi \rangle$ one defines 
$G(\vert \varphi \rangle ) = 1- \lambda^{2}(\vert \varphi \rangle )$ 
with 
\begin{align}
 \lambda^{2}(\vert \varphi \rangle ) = \underset{\vert \pi \rangle }{\text{max}} \, \vert \langle \pi \vert \varphi \rangle \vert^{2},
 \label{eq-geodef}
\end{align}
where the maximization runs over all product states $\vert \pi \rangle$ of 
the corresponding system. This quantity for pure states can be extended to
mixed states via the convex roof construction and is then a proper entanglement
monotone \cite{Wei2003}.

While computing $\lambda^{2}$ for a generic pure state is, in principle, difficult 
\cite{hardness_approximation_2008}, there is a simple see-saw iteration that can
be used \cite{est_ent_measures_2007, Streltsov2011, 
numerical_constructions_sperling_2018}. For a three-partite state $\vert \varphi 
\rangle$, the algorithm starts with a random product state $\vert a_{0} b_{0} c_{0} 
\rangle$. From this we can compute the non-normalized state $\vert \tilde{a} \rangle 
= \langle b_{0} c_{0} \vert \varphi \rangle$, and make the update $\vert a_{0} 
\rangle \mapsto \vert a_{1} \rangle = \vert \tilde{a} \rangle /\sqrt{\langle 
\tilde{a} \vert \tilde{a} 
\rangle}$. The procedure is repeated for the second qubit 
$\vert b_{0} \rangle$, starting in the product state $\vert a_{1} b_{0} c_{0} 
\rangle$. This is then iterated until one reaches a fixed point. Of course, 
this fixed point is not guaranteed to be the global optimum, in practice, however, 
this method works very well.

%%%%%%%%%%%%%%%%%%%%%%%%%%%%%%%%%%%%%%%%%%%%%%%%%%
%%%%%%%%%%%%%%%%%%%%%%%%%%%%%%%
%%%%%%%%%%%%%%%%%%%%%%%%%%%%%%%%%%%%%%%%%%%%%%%%%%

\section{Idea of the algorithm}  We present the algorithm for the case of 
three qubits. The generalization to arbitrary multiparticle systems is 
straightforward and is discussed in Appendix \ref{appendix-a}. As initial 
state $\vert \varphi \rangle$ we choose a random pure three qubit state. 
Then, we compute its closest product state $\vert \pi \rangle$ via the 
see-saw algorithm described above. We can assume without loss of generality 
that $\vert \pi\rangle = \vert 000 \rangle$. We write 
$\lambda =  \vert \langle \varphi \vert \pi \rangle \vert $ for 
the maximal overlap of $\vert \varphi \rangle$ with the set of all 
product states. Note that for a generic quantum state the closest 
product state is unique. 

The key idea is now to perturb the state $\ket{\varphi}$ in a way that 
the overlap with $\ket{\pi}$ decreases. If $\ket{\pi}$ is the unique
closest product state and the perturbation is small, one can then expect
that the overlap with {\it all} product states decreases. So, we consider 
the orthocomplement of $\vert \pi \rangle = \vert 000 \rangle$, that is, 
the complex subspace spanned by $ \vert 001 \rangle , \vert 010 \rangle, 
\vert 100 \rangle, \vert 011 \rangle , \vert 101 \rangle , \vert 110 \rangle , 
\vert 111 \rangle$. This subspace gives rise to a projection operator 
$\Pi = \openone - \ketbra{\pi}$ and we compute the best approximation of 
the state $\vert \varphi \rangle $ within this subspace, given by
$\vert \eta \rangle = \Pi \vert \varphi \rangle/\mathcal{M}$, where 
$\mathcal{M}$ denotes the normalization. Then, we 
shift the state $\ket{\varphi}$  in the direction of $\ket{\eta}$ by 
some small amount $\theta >0$. Hence the state update rule is given by 
\begin{align}
\label{eq:simpleupdate}
  \vert \varphi \rangle \mapsto \vert \Tilde{\varphi} \rangle  :=  \frac{1}{\mathcal{N}} ( \vert \varphi \rangle + \theta \vert \eta \rangle )
\end{align}
where $\mathcal{N}$ is a normalization factor. 

In the next step, we  calculate the best rank one approximation to 
$\vert \Tilde{\varphi} \rangle$. This process is iterated until the geometric 
measure is not increasing under the update rule (\ref{eq:simpleupdate}). In 
this case, one can reduce the step size or the algorithm terminates. 

One can directly check that the overlap with $\vert \pi \rangle$ is smaller 
for $\vert \Tilde{\varphi} \rangle$ than for $\vert \varphi \rangle$. Indeed 
one has 
\begin{align} 
\label{eq:crit}
    \vert \langle \pi \vert \Tilde{\varphi} \rangle \vert^{2} = \frac{1}{\mathcal{N}^{2}} \vert  \langle \pi \vert \varphi \rangle + \theta \langle \pi \vert \eta \rangle \vert^{2} = \frac{ \lambda^{2}}{\mathcal
    {N}^{2}} <  \lambda^{2}
\end{align}
since $\mathcal{N} >1$ if $\theta >0$. In fact, a much stronger statement holds:

\begin{observation}\label{thm:mainclaim}
For a generic quantum state $\vert \psi \rangle$ 
there always exists a $\Theta > 0$ such that the updated state 
$\vert \tilde{\psi} \rangle$ according to Eq.~\eqref{eq:simpleupdate} 
with step size $\theta<\Theta$ fulfills 
$G(\vert \psi \rangle) < G(\vert \tilde{\psi} \rangle )$. 
\end{observation}
Observation \ref{thm:mainclaim} provides the guarantee that the 
proposed algorithm yields a sequence of states with increasing 
geometric measure. 
The proof is given in Appendix~\ref{appendix-a} and comes with an interesting feature. It turns out that the proof does not rely on the particular product state structure of $\ket{\pi}$, so any figure of merit based on maximizing the overlap with pure states from some subset can be optimized with our method. This turns the algorithm into a powerful and tool with a universal applicability. Indeed, we adapt it to the dimensionality of entanglement (see Appendix~\ref{app:schmidt_rank}), the stabilizer rank (Appendix \ref{app:stabilizer}), matrix product states (Appendix~\ref{app:schmidt_rank}) as well as to the preparability in quantum networks~ (Appendix~\ref{app:triangle}). Remarkably, here novel states are found which are more distant to any network state as those known so far.

%%%%%%%%%%%%%%%%%%%%%%%%%%%%%%%%%%%%%%%%%%%%%%%%%
%%%%%%%%%%%%%%%%%%%%%%%%%
%%%%%%%%%%%%%%%%%%%%%%%%%%%%%%%%%%%%%%%%%%%%%%%%%

\section{Implementation}
Thanks to the update rule in Eq.~\eqref{eq:simpleupdate}, we can 
make use of advanced descent optimization algorithms in order to 
obtain faster convergence and higher robustness against local
optima~\cite{nesterov_book_2018,machine_learning_mehta_2019}. 
We have implemented a descent algorithm with momentum as well 
as the Nesterov accelerated gradient (NAG) \cite{nesterov1983}. 
The idea behind the momentum version is to keep track of the direction 
of the updates. More precisely, the update direction $\vert \eta_{n} \rangle$ 
in the $n$-th iteration will be a running average of the previously 
encountered updates $\vert \eta_{1} \rangle,...,\vert \eta_{n-1} \rangle$. 
For the NAG method, the update vector is, contrary to Eq.~(\ref{eq:simpleupdate}), 
evaluated at a point estimated from previous accumulated updates, and not at 
$\vert {\varphi} \rangle$.
A more detailed discussion and a comparison of the different methods can 
be found in Appendices~\ref{app:performance} and~\ref{app:random_starts}. The convergence of the algorithm in the case 
of qubits is also shown in Fig.~\ref{fig:convergence}.

After the algorithm has terminated, one obtains a tensor presented in a random basis.
In order to identify the state and to obtain a concise form, we have to find 
appropriate local basis for each party, such that the state becomes as simple 
as possible. This can be done by a suitable local unitary transformation, see 
Appendix~\ref{app:uniopt} for details.

%%%%%%%%%%%%%%%%%%%%%%%%%%%%%%%%%%%%%%%%%%%%%%%%%
%%%%%%%%%%%%%%%%%%%%%%%%%
%%%%%%%%%%%%%%%%%%%%%%%%%%%%%%%%%%%%%%%%%%%%%%%%%

\section{Results for qubits}
For two qubits, the maximally entangled state is the Bell state and our algorithm directly converges to 
this maximum. For the case of three qubits there are two different equivalence classes 
of genuine tripartite entangled states with respect to stochastic local operations and 
classical communication (SLOCC) \cite{three_qubits_ineuivalent_entangled_duer_2000}, namely 
$\vert \text{W} \rangle = (\vert 001 \rangle + \vert 010 \rangle + \vert 100 \rangle)/\sqrt{3}$ 
and $\vert \GHZ \rangle = (\vert 000 \rangle + \vert 111 \rangle)/\sqrt{2}$. While 
$G(\vert \GHZ \rangle) = 1/2$ one has $G(\vert \text{W} \rangle) = 5/9$, what turns out to be the 
maximizer among all tripartite states~\cite{geomeasure_three_qubit_characterization_2009}. Indeed, after $150$ iterations, the algorithm yields 
the W state. 
It should be noted, that in this case the maximizer belongs to the family of symmetric 
states. Operationally, the W state is the state with the maximal possible bipartite 
entanglement in the reduced two-qubit states~\cite{three_qubits_ineuivalent_entangled_duer_2000}.

For four qubits, the algorithm yields after $300$ iterations the state
\begin{align} \label{eq:4qubits}
\ket{\tilde{\text{M}}} = \frac{1}{\sqrt{3}}
(\vert \GHZ \rangle + e^{2 \pi i/3} \vert \GHZ_{34}\rangle + e^{4 \pi i/3} \vert \GHZ_{24} \rangle)
\end{align}
where $\vert \GHZ_{ij} \rangle$ means a four-qubit GHZ state where a bit flip is applied at 
party $i$ and $j$. Note that the phases form a trine in the complex plane and that the state 
is a phased Dicke state~\cite{entanglement_structure_factor_krammer_2009}. This state can be shown to be LU-equivalent 
to the so called Higuchi-Sudbery or M state~\cite{higuchi2000, gour2010}, which appears 
as maximizer of the Tsallis $\alpha$-entropy in the reduced two-particle states for 
$0< \alpha<2$.  Note that this state is not symmetric with respect to permutations of 
the parties. Similar to the W state, the entanglement of the state in Eq.~\eqref{eq:4qubits} 
appears to be robust, that is, uncontrolled decoherence of one qubit does not completely 
destroy the entanglement of the remaining qubits~\cite{higuchi2000}.

For five qubits the algorithm converges to a state $\vert G_{5} \rangle$, that can be 
identified with the ring cluster state (a $5$-cycle graph state) yielding a geometric 
measure of $0.86855 \approx \frac{1}{36} (33- \sqrt{3})$. The state  $\vert G_{5} \rangle$
appears in the context of the five-qubit error correcting 
code~\cite{5qubit_error_correction_2001}. Some basic facts concerning cluster and graph states
are given in Appendix~\ref{app:ame_and_graphs}. Similarly, for six qubits we obtain a graph state 
$\vert G_{6} \rangle$ with a measure of $0.9166 \approx \frac{11}{12}$. This is again
connected to quantum error correction, indeed, both states $\vert G_{5} \rangle$
and $\vert G_{6} \rangle$ are AME states, see also below. For seven qubits, we 
find a numerical state with maximally mixed two-body marginals, where the spectra 
of the three-body marginals are all the same. This motivates to introduce the class of maximally marginal symmetric states (MMS), as a natural extension of notion of AME states.

%%%%%%%%%%%%%%%%%%%%%%%%%%%%%%%%%%%%%%%%%%%%%%%%%%%%%%%%%%%%%%%%%%%%%%%%%%%%%%%%%%%%%%%%
\begin{table}[t]
    \centering
    \begin{tabular}
    {|c|c|c|c|}
    \hline
    $n$ & $G_{\text{max}}^{\text{symm}}$ & $G_{\text{max}}$ & $\vert \varphi \rangle_{\text{max}}$ \\ 
    \hline \hline
         2 & $1/2$ & $1/2$ & $\vert \psi^{-} \rangle$\\ \hline
         3 & $0.5555 \approx 5/9 $ & $0.5555 \approx5/9 $& $\vert \text{W} \rangle$\\ \hline
         4 & $0.6666 \approx 2/3$ & $0.7777 \approx 7/9 $& $\vert \text{M} \rangle$ \\ \hline 
         5 & $\approx 0.7006$ & $0.8686 \approx (1/36)(33-\sqrt{3})  $& $\ket{G_{5}}$\\ \hline 
         6 & $0.7777 \approx 7/9 $ & $ 0.9166 \approx 11/12 $& $\ket{{G}_{6}}$\\ \hline 
         7 & $\approx 0.7967$ & $\geq 0.941$ & $\text{MMS}(7,2)$ \\ \hline
    \end{tabular}
    \caption{Maximally entangled states found by the algorithm for systems between 
    two and seven qubits. Here $\vert \varphi \rangle_{\text{max}}$ refers to the 
    state found by algorithm and $G_{\text{max}}$ denotes the geometric measure 
    of the corresponding state. $G_{\text{max}}^{\text{symm}}$ denotes the maximal entanglement among symmetric states, as shown in~\cite{aulbach2010}.}
    \label{tab:my_label}
\end{table}
%%%%%%%%%%%%%%%%%%%%%%%%%%%%%%%%%%%%%%%%%%%%%%%%%%%%%%%%%%%%%%%%%%%%%%%%%%%%%%%%%%%%%%%%

%%%%%%%%%%%%%%%%%%%%%%%%%%%%%%%%%%%%%%%%%%%%%%%%%
%%%%%%%%%%%%%%%%%%%%%%%%%
%%%%%%%%%%%%%%%%%%%%%%%%%%%%%%%%%%%%%%%%%%%%%%%%%
\section{Relation to AME states} 
One calls a multiparticle state AME, if it is a maximally entangled state for 
any bipartition. For bipartite entanglement measures, it is well known that 
pure states with maximally mixed marginals are the maximally entangled states  \cite{max_entangled_states_pasquinucci_1998,monotones_vidal_2000, 
higuchi2000}. Consequently, an $n$-partite pure multi-qudit state is AME if 
all reductions to $\lfloor \frac{n}{2} \rfloor$ parties are maximally mixed, 
such a state is then denoted by $\text{AME}(n,d)$. Interestingly, not for all 
values of $n$ and $d$ AME states exist and quest for AME states is a central
problem in entanglement theory \cite{5open_question_horodecki_2022}. The non-existence of AME states was first encountered for four qubits \cite{higuchi2000}, 
but interestingly, the M state in Eq.~\eqref{eq:4qubits} that maximizes
the geometric measure can be viewed as the best possible replacement, since
the one-body marginals are maximally mixed and all $2$-body marginals, albeit 
not being maximally mixed have the same spectrum \cite{gour2010}. For five and
six qubits the states found by our algorithm are just the known AME states, 
see also Appendix \ref{app:ame_and_graphs}. For $n\geq 7$ no $\text{AME}(n,2)$ state exist \cite{huber2017,scott2004}. A potential approximation $\vert F \rangle$ 
to the AME state of seven qubits has been identified \cite{huber2017}, which 
is a graph state corresponding to the Fano plane.  This state has a measure 
of $G(\vert F \rangle ) = {15}/{16} = 0.9375$, thus smaller as the measure 
of the MMS state found by the algorithm.

While for certain values of $n$ and $d$ no AME state exists, it appears 
that multiple AME states can exist for other choices. In fact, it has 
been shown that those states can be even SLOCC inequivalent~\cite{raissi2020,burchardt2020}. 
It is in general a difficult problem to decide whether two states belong to the 
same SLOCC class, but for the case of AME states it can be drastically simplified 
using our algorithm in combination with the Kempf-Ness theorem~\cite{kempf_ness_thm_1979}. 
First, notice that states with maximally mixed $1$-body marginals belong to the so-called class of 
critical states~\cite{local_manipulation_gour_2011}. The theorem then assures that two critical states belong 
to the same SLOCC class if and only if they are LU equivalent. Hence,
if the geometric measure of those states differs, it already implies
SLOCC inequivalence.

%%%%%%%%%%%%%%%%%%%%%%%%%%%%%%%%%%%%%%%%%%%%%%%%%
%%%%%%%%%%%%%%%%%%%%%%%%%
%%%%%%%%%%%%%%%%%%%%%%%%%%%%%%%%%%%%%%%%%%%%%%%%%

\section{Systems of higher dimensions}
For the bipartite case the generalized Bell states $\vert \phi_{d} \rangle =
\sum_{j=0}^{d-1} \vert jj \rangle / \sqrt{d}$ are maximally entangled with 
$G(\vert \phi_{d} \rangle) = 1 - 1/d$. We find that for $2\leq d \leq 10$ 
the algorithm yields the corresponding state $\vert \psi_{d} \rangle$ with 
high fidelity, and that the number of iterations needed until convergence 
appears to be only weakly dependent on $d$, see also Appendix~\ref{app:performance}.

In the three-qutrit case 
we obtain the total antisymmetric state $\vert \Psi_3 \rangle$, 
given by 
\begin{align}
\vert \Psi_{3} \rangle =  \frac{1}{\sqrt{6}}
(\vert 012 \rangle +  \vert 201 \rangle + \vert 120 \rangle 
- \vert 210 \rangle - \vert 102 \rangle - \vert 021 \rangle)
\end{align}
In general, antisymmetric states $\vert \Psi_{n} \rangle $ can be constructed 
for all $n$-partite n-level systems and their geometric measure can be easily
computed analytically as
$\frac{n!-1}{n!}$~\cite{entanglement_asymm_states_virmani_2008}. In the 
particular case $n=3$ we obtain $G = \frac{5}{6} \approx 0.8333$. Note 
that $\vert \Psi_{3} \rangle$ is an AME state.  

More generally, in the tripartite case a procedure is known to construct 
AME states for arbitrary $d$~\cite{goyeneche2018}. This leads to
\begin{align} 
\label{eq:AME3d}
\text{AME}(3,d) \sim \sum_{i,j =0}^{d-1} \vert i \rangle \vert j \rangle \vert i+j \rangle, 
\end{align}
where $i+j$ is computed modulo $d$. Note that the state 
$\text{AME}(3,3)$ constructed according to Eq.~\eqref{eq:AME3d} only 
has a measure of $2/3$. 

For the case of three ququads our algorithm gives insights into the AME 
problem. First, the $\text{AME}(3,4)$ state corresponding to Eq.~(\ref{eq:AME3d})
has a geometric measure of $G= 0.75$. However, our algorithm yields a state given 
by 
\begin{align} 
\label{eq:AME34}
\begin{split}
    \vert \phi_{3,4}  \rangle = \frac{1}{2 \sqrt{2}} ( \vert 022 \rangle + \vert 033 \rangle + \vert 120 \rangle + \vert 1 31 \rangle \\ 
    + \vert 212 \rangle + \vert 203 \rangle + \vert 310 \rangle + \vert 301 \rangle  )
    \end{split}
\end{align}
with $G(\vert \phi_{3,4} \rangle ) = 7/8 = 0.875$. After applying 
local unitaries, this state can be seen as arising from three Bell pairs distributed between three parties in a triangle like configuration. In addition, $\vert \phi_{3,4} \rangle$ is an AME state, 
i.e., all one-party marginals are maximally mixed. As the geometric 
measure of the states $\vert \phi_{3,4} \rangle$ and $\text{AME}(3,4)$ differs, they belong to different SLOCC classes.

In the case of four qutrits the algorithm converges to a state with 
a geometric measure of $0.888 \approx \frac{8}{9}$. This state can be identified 
to be the $\text{AME}(4,3)$ state given by \cite{goyeneche2018}
\begin{align}
\begin{split}
     \text{AME}(4,3) = & \frac{1}{3} ( \vert 0000 \rangle + \vert 0112 \rangle + \vert 0221 \rangle + \vert 1011 \rangle  + \vert 1120 \rangle \\  & + \,  \vert 1202 \rangle  +  \vert 2022 \rangle + \vert 2101 \rangle + \vert 2210 \rangle) 
    \end{split}
\end{align}

For four ququads, the algorithm converges to the antisymmetric state $\vert \Psi_{4} \rangle$, yielding a measure of $\frac{23}{24} \approx 0.9583$. Interestingly, while being $1$-uniform, this state is not AME. The so far only known $\text{AME}(4,4)$, a graph state \cite{AME_qudit_helwig_2013,burchardt2020}, yields a measure of $\frac{15}{16} = 0.9375$. 
Finally, the recently found $\text{AME}(4,6)$~\cite{rather2021} is not maximally entangled 
with respect to the geometric measure, i.e., the algorithm finds states yielding a higher 
geometric measure.

Similarly, there exists a general procedure to construct $\text{AME}(5,d)$ 
states given by~\cite{nonbinary_qc_rains_1999,goyeneche2018}
\begin{align}
    \text{AME}(5,d) \sim \sum_{i,j,l=0}^{d-1} \omega^{il} \vert i \rangle \vert j 
    \rangle \vert i+j \rangle \vert l+j \rangle \vert l \rangle, 
\end{align}
where $\omega = e^{2 \pi i /d}$. 
In the case of a three-dimensional system, the algorithm converges to the 
$\text{AME}(5,3)$ state yielding a measure of approximately $0.96122$. 
Finally, we have $G(\text{AME}(5,4)) = \frac{31}{32} = 0.96875$. Here the 
algorithm yields a state $\vert \phi_{5,4} \rangle$ with a larger geometric 
measure, in particular $G(\vert \phi_{5,4} \rangle) > 0.975$.  
However, here we cannot identify 
a closed expression of the state. The numerical result suggests that 
the maximizer is again an AME state.

%%%%%%%%%%%%%%%%%%%%%%%%%%%%%%%%%%%%%%%%%%%%%%%%%
%%%%%%%%%%%%%%%%%%%%%%%%%
%%%%%%%%%%%%%%%%%%%%%%%%%%%%%%%%%%%%%%%%%%%%%%%%%

\section{Maximally entangled subspaces}
One can extend our method such that it also applies to 
subspaces. More precisely, we want to construct an orthonormal 
basis for a subspace $V$ such that the least entangled state 
in $V$ is as entangled as possible in comparison with all 
other possible subspaces. Note that this notion differs from 
the concept of genuine entangled subspaces (GES) \cite{demian2018}, 
where all states within the subspace have to be genuine entangled. 
However, our algorithm can readily be modified in order to search 
for GES with maximal genuine multipartite entanglement. 

We will explain the idea of the algorithm for a two-dimensional 
subspace of three qubits. 
First, we choose a two-dimensional subspace randomly, which can be 
described by a projector of the form 
$P = \vert v \rangle \langle v \vert + \vert w \rangle \langle w \vert$, 
where $\langle v \vert w \rangle =0$. Next, we compute the best rank-one 
approximation to $P$ given by taking the argmax of
$\text{sup}_{abc} \, \text{Tr}[\vert abc \rangle \langle abc \vert P]$. 
This can be done with the iteration described after Eq.~(\ref{eq-geodef}).
Again we call the optimizer $\vert \pi \rangle$. More 
generally, the optimization yields the best product state approximation 
to the state in the range of $P$ that is least entangled.
In particular, this implies that if $im(P)$ contains a product state, the 
assigned geometric measure will be zero. 
Then we compute the eigenvectors corresponding to the two largest eigenvalues 
of the operator $P- \theta \vert \pi \rangle \langle \pi \vert$ that we 
will call $\vert v_{1} \rangle $ and $ \vert v_{2} \rangle$. Clearly, this 
algorithm reduces to the one from the previous sections, if we choose 
the rank of the projector to be one.

It is known \cite{partha2004}, that the maximal dimension of a 
subspace of an $n$-partite qudit system which contains no product 
state is given by $\text{dim}(V) = d^{n} - dn +n-1$. Consequently, for two qubits, the maximal entangled 
subspace is of dimension one and spanned by a Bell state. Indeed, 
if we apply our algorithm to larger subspaces, the measure we 
assign to this subspaces stays zero.

For three qubits, the algorithm converges to the W state as one basis vector 
and to the state
\begin{align}
 \vert V \rangle = \frac{1}{\sqrt{3}}
 \big[\vert 011 \rangle + e^{\frac{2 \pi}{3} i} \vert 110 \rangle +e^{\frac{4 \pi}{3} i} \vert 101 \rangle 
 \big]
\end{align}
as the other. The subspace spanned by $\vert W \rangle$ and $\vert V \rangle$ 
then has a remarkable property: all states are maximally entangled, yielding 
the same geometric measure as the W state. This has potential applications
in information processing: In this subspace, qubit states may be encoded and 
then any set of these states is difficult to discriminate by local means
\cite{Hayashi2006}.

Concerning higher dimensions, we computed the maximally entangled subspace of 
dimension two for two qutrits. 
Here we obtain an embedded Bell state 
$\ket{\chi_1} = (\vert 01 \rangle - \vert 10 \rangle)/\sqrt{2}$ 
and the state
\begin{align}
\ket{\chi_2} =\frac{1}{\sqrt{14}} 
\big[ 
\vert 20 \rangle + \vert 02 \rangle + \sqrt{6} (\vert 21 \rangle + \vert 12 \rangle)
\big]
\end{align}
which can also be seen as the superposition of two Bell pairs. It turns out 
that the least entangled state within this subspace has a geometric measure of 
$1/2$. Again, applications can be envisaged, as any state in this subspace has at 
least one ebit of entanglement.

%%%%%%%%%%%%%%%%%%%%%%%%%%%%%%%%%%%%%%%%%%%%%%%%%
%%%%%%%%%%%%%%%%%%%%%%%%%
%%%%%%%%%%%%%%%%%%%%%%%%%%%%%%%%%%%%%%%%%%%%%%%%%

%%%%%%%%%%%%%%%%%%%%%%%%%%%%%%%%%%%%%%%%%%%%%%%%%
%%%%%%%%%%%%%%%%%%%%%%%%%
%%%%%%%%%%%%%%%%%%%%%%%%%%%%%%%%%%%%%%%%%%%%%%%%%

\section{Conclusion}
In this work we have presented an iterative method for the computation of maximally resourceful quantum states. We provided a convergence analysis and showed that in each step the resourcefulness of the iterates increase. We illustrated our approach for the special case of the geometric measure, allowing us to identify interesting quantum states, discover novel AME states, and characterize highly entangled subspaces which may be useful for information processing. We further demonstrated the universality of the algorithm for various other quantifiers, yielding novel forms of correlations in the triangle network.

Concerning further research, our results also suggest a variety of avenues for further theoretical exploration. Can the algorithm be used to find new AME states for cases where the existence is still open,  e.g., for systems consisting of more then five six-dimensional systems, or to find new SLOCC inequivalent AME states? In particular, we have numerical evidence that there exists a second AME state (besides \cite{rather2021}) for the case of four six-dimensional systems. 
From a mathematical perspective, 
the algorithm can give insights into the structure of tensor spaces and could offer intuition to solve open problems concerning the asymptotics 
of tensor norms~\cite{aubrun2017}. 
%Finally, 
%the algorithm can be applied to other quantifiers of 
%quantumness, thus open opportunities in many different %fields of physics. 

%%%%%%%%%%%%%%%%%%%%%%%%%%%%%%%%%%%%%%%%%%%%%%%%%
%%%%%%%%%%%%%%%%%%%%%%%%%
%%%%%%%%%%%%%%%%%%%%%%%%%%%%%%%%%%%%%%%%%%%%%%%%%

\begin{acknowledgments}
We thank 
Guillaume Aubrun,
Matthias Christandl,
Mariami Gachechiladze,
Felix Huber,
H. Chau Nguyen, 
Thorsten Raasch,
Ren\'e Schwonnek, 
Liqun Qi,
Franz Josef Stoiber,
Julio de Vicente, 
Nikolai Wyderka,
Zhen-Peng Xu,
Guofeng Zhang, 
and
Karol $\dot{\text{Z}}$yczkowski
for inspiring discussions. 
The University of Siegen is 
kindly acknowledged for enabling our computations through the \texttt{OMNI} cluster. 
This work was supported by the Deutsche Forschungsgemeinschaft (DFG, German Research 
Foundation, project numbers 447948357 and 440958198), the Sino-German Center 
for Research Promotion (Project M-0294), the ERC (Consolidator Grant 683107/TempoQ) 
and the German Ministry of Education and Research (Project QuKuK, BMBF Grant
No. 16KIS1618K).
JS acknowledges support from the House of Young Talents of the University of Siegen.
\end{acknowledgments}

\appendix

%\onecolumngrid

\section{Proof of Observation 1}
\label{appendix-a}
Let us first recapitulate the algorithm and fix the notation. Let $\vert \psi \rangle$ 
be the randomly drawn initial state and denote by $\vert \pi \rangle $ its best product 
state approximation (BPA), that is, 
\begin{align} \label{min}
    \vert \pi \rangle = \text{argmin} \lbrace \vert \pi \rangle \, 
    : \, 1- \vert \langle \varphi \vert \pi \rangle \vert^{2} \, 
    : \, \, \vert \pi \rangle \, \, \text{product state} \rbrace.
\end{align}
In particular, we can choose $\vert \pi \rangle$ such that $\lambda := \langle \pi \vert \psi \rangle > 0$. 
This allows for the computation of the update direction, 
$\vert \eta \rangle := (\Pi \vert \psi \rangle)/\mathcal{M}$, where 
$\Pi:= \openone - \vert \pi \rangle \langle \pi \vert$ is the projector 
onto the orthocomplement of the span of the BPA $\vert \pi \rangle$ and $\mathcal{M}$ 
a normalization such that $\langle \eta \vert \eta \rangle = 1$. It follows directly 
that $\mathcal{M}^2 = \langle \psi \vert P \vert \psi \rangle = 1 - \lambda^{2}$, hence 
$\mathcal{M} = \sqrt{1-\lambda^{2}}$. For a fixed step size $\theta >0$, the update 
rule is given by 
\begin{align} 
\label{app:nice_update}
\vert \psi \rangle \mapsto \vert \tilde{\psi} \rangle 
:= \frac{1}{\mathcal{N}} (\vert \psi \rangle + \theta \vert \eta \rangle).
\end{align}
where $\mathcal{N}$ is again a normalization. In the following, we will frequently 
make statements about generic states. In these cases, we require the assumptions
that a state is not a product state and that the BPA $\ket{\pi}$ is unique up to 
a phase. 

The first step in order to prove monotonicity of the algorithm is to show that 
small variations in the initial state can only lead to small variations in the BPA. 
As we will see, this follows from the more general observation that under certain 
conditions the value $y_0$, where a function $f(x_0,y)$ assumes its minimum 
(for a given $x_0$), depends continuously on $x_0$. Further, by virtue of the 
canonical embedding,  we can identify any $\vert \varphi \rangle \in \mathbb{C}^{n} $ 
with a $\vert \Tilde{\varphi} \rangle  \in \mathbb{R}^{2n}$ and consequently 
we can omit the absolute in Eq.~\eqref{min}. Then we have:

\begin{lemma} 
\label{app:lemma_1}
Let $X,Y$ be compact and $f : X \times Y \rightarrow \mathbb{R}$ 
be uniformly continuous. Further, suppose that for 
$x_{0} \in X$ the value $y_{0} := {\text{argmin}}_{y \in Y}  f (x_{0},y)$ 
is unique. Then for all $\varepsilon > 0$ there exists $\delta >0$ such that 
for all $x \in U_{\delta} (x_{0} )$ we have 
${\text{argmin}}_{y \in Y} \, f (x,y) \subset U_{\varepsilon}(y_{0})$, where 
$U_{\delta(x_{0})}$ and $U_{\varepsilon}(y_{0})$ denote vicinities of $x_0$ and 
$y_0$, respectively. In other words, the function $\text{argmin}$ is 
continuous in $x_0$.
\end{lemma}

\begin{proof}
For the given $\varepsilon$ we can split the set $Y$ in the vicinity 
$U_{\varepsilon} (y_{0})$ and its complement $\overline{U}_{\varepsilon} (y_{0})$.
In particular we have 
\begin{align}
   f(x_{0}, y_{0} ) &= 
    \min_{y \in U_{\varepsilon}(y_{0})}{f (x_{0} ,y )} 
    \nonumber
    \\
    & < 
    \min_{y \in \overline{U}_{\varepsilon} (y_{0})}{f (x_{0} , y )} 
    =: f ( x_{0}, \tilde{y}_{0}),
\end{align}
that is, $\tilde{y}_{0}$ denotes the value where the minimum in 
$\overline{U}_{\varepsilon} (y_{0})$ is assumed. Let us denote the 
difference between the function values as
\begin{align}
\xi = f ( x_{0}, \Tilde{y}_{0}) -  f(x_{0}, y_{0} ) > 0.
\end{align}

By the uniform continuity we can choose $\delta>0$ such that for all 
$\tilde{x} \in \mathbb{R}^{n}$ with $\Vert \tilde{x}- x_{0} \Vert < \delta$ 
and for all $y$ we have 
\begin{align}
\vert f ( \tilde{x} , y) - f( {x}_0, y )\vert < \frac{\xi}{2}.
\end{align}
Then we have 
\begin{align}
f(\tilde{x}, y_0) < f(x_0, y_0) + \frac{\xi}{2},
\end{align}
but for all $y\in \overline{U}_{\varepsilon} (y_{0})$
\begin{align}
f(\tilde{x}, y) > f(x_0, \tilde{y}_0) - \frac{\xi}{2} > f(x_0, y_0) + \frac{\xi}{2},
\end{align}
which implies that the minimum of $f(\tilde{x}, y)$ lies in the vicinity 
$U_{\varepsilon} (y_{0})$. 
\end{proof}

\begin{corollary}\label{app:corr_1}
Let $\vert \varphi \rangle $ be a pure quantum state and suppose that 
its BPA $\vert \pi \rangle $ is unique. Then, for all $\tau >0$, there 
exists a $\xi >0$ such that the BPA $\vert \Tilde{\pi} \rangle $ of  
$\vert \Tilde{\varphi} \rangle \in \mathcal{U}_{\xi} ( \vert \varphi \rangle ) $ 
lies in $\mathcal{U}_{\tau} (\vert \pi \rangle )$.
\end{corollary}
\begin{proof}
The function $f(x,y) := \vert \langle x,y \rangle \vert^{2}$ is 
continuous on $\mathbb{R}^{2n} \times \mathbb{R}^{2n}$. Further, 
the space $B_{1} := \lbrace x \in \mathbb{R}^{2n} \, : \, \vert \vert x \vert \vert =1 \rbrace$ 
is compact and thus also $M:=B_{1} \times B_{1}$. Then, by the Heine-Cantor 
theorem~\cite{continuity}, $f$ is uniformly continuous on $M$. Since we 
assume $\vert \pi \rangle$ to be unique, we can apply Lemma~\ref{app:lemma_1}, 
which guarantees for all $\tau > 0$ the existence of $\xi >0$ such that 
$\vert \tilde{\pi} \rangle \in U_{\tau}(\vert \pi \rangle)$ for 
$\vert \tilde{\psi} \rangle \in U_{\xi}(\vert \psi \rangle)$. 
\end{proof}

According to Eq.~\eqref{app:nice_update}, the updated state needs a renormalization given by $\mathcal{N}= \mathcal{N}(\vert \psi \rangle , \theta)$. The next ingredient for the proof of the main result is a Lemma
that gives later an upper approximation of the function $1/\mathcal{N}$. 

\begin{lemma} \label{app:lemma_2}
There exists $C>0$ such that for all $q \in [0,1]$ and $x>0$ we have 
\begin{align} \label{eq:norm_estimate}
\frac{1}{\sqrt{1+2qx+x^2}} < 1- qx + Cx^{2}.
\end{align}
More precisely, the above inequality holds for all $C \geq 3$. 
\end{lemma}
\begin{proof}
As both sides of Eq.~\eqref{eq:norm_estimate} are positive, we can square them 
such that the inequality remains true. 
This yields the equivalent inequality
\begin{align}
\begin{split}
    0  < & [2C-3q^2 +1]x^2 + 2q[C-1+q^2]x^3 + \\
    & + [C^2 - C(4q^2-2) +q^2]x^4 + 2q[C^2-C]x^5 + C^2x^6 \\
    & =: f_{2}x^2 + f_{3}x^3+f_{4}x^4+f_{5}x^5+f_{6}x^6
    \end{split}
\end{align}
Now observe that for each of the constants $f_{k} = f_{k}(C)$ there is a $C_{k} >0$ such that 
$f_{k}(C) \geq 0$ for all $C \geq C_{k}$. Indeed, we have 
$C_{2} := \text{max} \lbrace (1/2)(3q^2-1), 0 \rbrace$,
$C_{3} := 1-q^2$, 
$C_{5} = C_{6} := 1$. 
The choice of $C_{4}$ depends on whether $q^2 \geq 1/2$ or not. If we denote $\alpha = \vert 4 q^2-2 \vert$, we obtain for the case $q^2<1/2$ that $C^2+ \alpha C + q^2 >0$, what is trivially fulfilled for any $C\geq 1$. If  $q^2 > 1/2$, we need $C^2-\alpha C +q^2 >0$. But $C^2-\alpha C +q^2 \geq C^2 - \alpha C = C(C-\alpha)>0$, we obtain $C >\alpha$.  
In general, we have $C_{4} := \text{max} \lbrace 1, \vert 4q^2-2 \vert \rbrace$. This implies that 
for $C\geq \tilde{C} := \text{max} \, \lbrace  C_{k} \, \vert \, k=2,...,6 \rbrace$ all coefficients are positive. Hence $0< f_{2}x^2 $ implies $0 <f_{2}x^2 + f_{3}x^3+f_{4}x^4+f_{5}x^5+f_{6}x^6 $ if $x>0$. Consequently, it is sufficient to only consider the problem $0 < f_{2}x^{2}$, what is true for $C \geq C_{7} := (1/2)(3q^2-1)$. Hence, choosing  $C \geq \text{max} \lbrace \Tilde{C},C_{7} \rbrace$ yields the claim. Taking the maximum over all $C_{k}$ with respect to $q \in [0,1]$ yields that $C>2$. 
\end{proof}

\begin{theorem}
Let $\vert \psi \rangle$ be a generic pure quantum state. Then there exists $\Theta >0$ 
such that for the updated state $\vert \tilde{\psi} \rangle$ according to 
Eq.~\eqref{app:nice_update} with step-size $\Theta > \theta > 0$ we have 
$G(\vert \psi \rangle ) < G(\vert \tilde{\psi} \rangle)$. 
\end{theorem}

\begin{proof}
Let us start with some step-size $\theta_{0} > 0$ that we will choose in the end 
appropriately and consider
\begin{align}
\vert \tilde{\psi} \rangle = \frac{1}{\mathcal{N}} 
(\vert \psi \rangle + \theta_{0} \vert \eta \rangle). 
\label{eq-app-1}
\end{align}  
It is important to note that $\vert \eta \rangle$ is a normalized state, that is,
$\vert \eta \rangle = 1/(\sqrt{1-\lambda^2}) (\openone - \vert \pi \rangle \langle \pi \vert) 
\vert \psi \rangle$. This yields $\langle \psi \vert \eta \rangle = \sqrt{1-\lambda^2}$. 

The BPA $\vert \tilde{\pi} \rangle $ of $\vert \tilde{\psi}\rangle$ can be parameterized 
using the old product state, i.e., $\vert \tilde{\pi} \rangle = \sqrt{1-\delta^2} \vert \pi \rangle + \delta \vert \chi \rangle$, for a normalized, appropriately chosen $\vert \chi \rangle$ and $\delta > 0.$ 
Using $\langle \pi \vert \eta \rangle =0$ 
and 
$|\langle \chi \vert \eta \rangle| \leq 1$
we obtain
\begin{align}
    \tilde{\lambda} := &
    | \langle \tilde{\pi} \vert \tilde \psi \rangle| 
    = \frac{1}{\mathcal{N}} 
    \big|
    \braket{\tilde{\pi}}{\psi} + \theta_0 \delta \braket{\chi}{\eta}  
    \big| 
    \nonumber
    \\
   & \leq
    \frac{1}{\mathcal{N}} 
    \big(
    |\braket{\tilde{\pi}}{\psi}|
    + 
    |\theta_0 \delta \braket{\chi}{\eta}|  
    \big)
    \leq
    \frac{1}{\mathcal{N}} (\lambda + \delta \theta_{0}).
\end{align}
Using that $\mathcal{N}=\sqrt{1+2\theta_0 \sqrt{1-\lambda^2}+ \theta_0^2}$  and 
Lemma~\ref{app:lemma_2} there exists $C>0$ such that 
\begin{align}
    \Tilde{\lambda} & < (1- \theta_{0} \sqrt{1-\lambda^2} + C \theta_{0}^{2}) (\lambda + \delta \theta_{0})
    \nonumber \\
    & = \lambda + \delta \theta_{0} - \lambda \sqrt{1-\lambda^2} \theta_{0} 
    + \theta_0^2\big[C (\lambda + \delta \theta_0) - \delta \sqrt{1-\lambda^2}\big]
    \nonumber \\
    & = \lambda + \theta_0 (\delta - \lambda \sqrt{1-\lambda^2}) + \mathcal{O}(\theta_0^2).
    \label{eq-app-2}
\end{align}
Note that $\lambda \sqrt{1-\lambda^2}>0$, since $0\neq\lambda\neq1$. This comes from the fact that
a generic state is not a product state with $\lambda=1$ and any state has at least some overlap
with some product state. So, if $\delta < \lambda \sqrt{1-\lambda^2}$ we have 
$\tilde{\lambda} < \lambda$ for suitably small $\theta_0$.

It remains to show that we can guarantee that $\delta$ obeys this condition. We start with a given
value of $\lambda$ and consider a number $0 < \delta_1 < \lambda \sqrt{1-\lambda^2}$. According to 
 Corollary~\ref{app:corr_1}, we can find a $\theta_1>0$ such that 
 $\vert \tilde{\pi} \rangle \in U_{\delta_{1}} (\vert \pi \rangle)$ 
 if $\vert \tilde{\psi} \rangle \in U_{\theta_1} (\vert \psi \rangle)$. 
 Then, this gives us an upper bound on $\theta_0$ for Eq.~(\ref{eq-app-1}), 
 so that the resulting $\delta<\delta_1$ in Eq.~(\ref{eq-app-2}) is small 
 enough to guarantee a negative slope for the linear term. Still, 
 $\tilde{\lambda} < \lambda$ is not guaranteed, due to the 
 $\mathcal{O}(\theta_0^2)$ term in Eq.~(\ref{eq-app-2}). But, 
 for the given values of $\lambda, \delta_1$ and $C$, we can 
 also compute from Eq.~(\ref{eq-app-2}) a second threshold $\theta_2$, 
 which guarantees $\tilde{\lambda} < \lambda$. Then we can take finally
 in the statement of the Theorem $\Theta = \min\{\theta_1, \theta_2\}$
 and the proof is complete.
 \end{proof}
It is remarkable that in this proof the fact that $\ket{\pi}$ is a product
state was never used. So, the algorithm can also be used if the overlap with 
states from some other subset of the (pure) state space shall be minimized. In particular, the given subset must not necessarily be presented in form of a smooth submanifold, but could also be discrete (see Appendix \ref{app:stabilizer}).
Further, the algorithm also applies to the case where each particle may have a different degree of freedom, e.g., $\mathbb{C}^{2} \otimes \mathbb{C}^{3} \otimes \mathbb{C}^{5}$. For this, only the subroutine of computing the best rank-$1$ approximation has to be modified.

\section{Algorithm for Haar-random starting points}
\label{app:random_starts}
In this section we first discuss how (pure) quantum states can be sampled according to the Haar measure. Thereupon we analyse how the geometric measure is distributed for small multi-qubit systems and show how the trajectories of the algorithm behave for different (Haar random) starting points.
 
As we want to sample from the set of pure states, i.e., $\vert \varphi \rangle \in (\mathbb{C}^{d})^{\otimes n}$, we can identify the state space with the unit sphere of $\mathbb{C}^{d^{n}}$. We will denote the $n$ dimensional complex unit sphere by $\mathbb{S}^{n-1}$. Further, we say that a complex-valued random variable $X$ is standard normal distributed, denoted by $X \sim \mathcal{N}(0,1)$, if the real-valued random variables $\Re(X)$ and $\Im (X)$ are independent and standard normal distributed (mean value 0 and variance 1). In order to sample a point uniformly at random according to Haar measure on $\mathbb{S}^{n-1}$, see also Ref.~\cite{random_matrices_karol_2003} for more detailed explanations,
one might consider a sequence $(X_{1}, X_{2},..., X_{n})$ of normal distributed independent random variables $X_{k} \sim \mathcal{N}(0,1)$. Then, by the property of Gaussians, the vector $(X_{1}, X_{2},...,X_{n}) \in \mathbb{C}^{n}$ is a rotationally invariant $n$-dimensional Gaussian. Normalizing the vector $(X_{1}, X_{2},...,X_{n})$ yields a uniformly random point on $\mathbb{S}^{n-1}$.

\begin{figure}[t]
\centering
    \includegraphics[width=0.9\columnwidth]{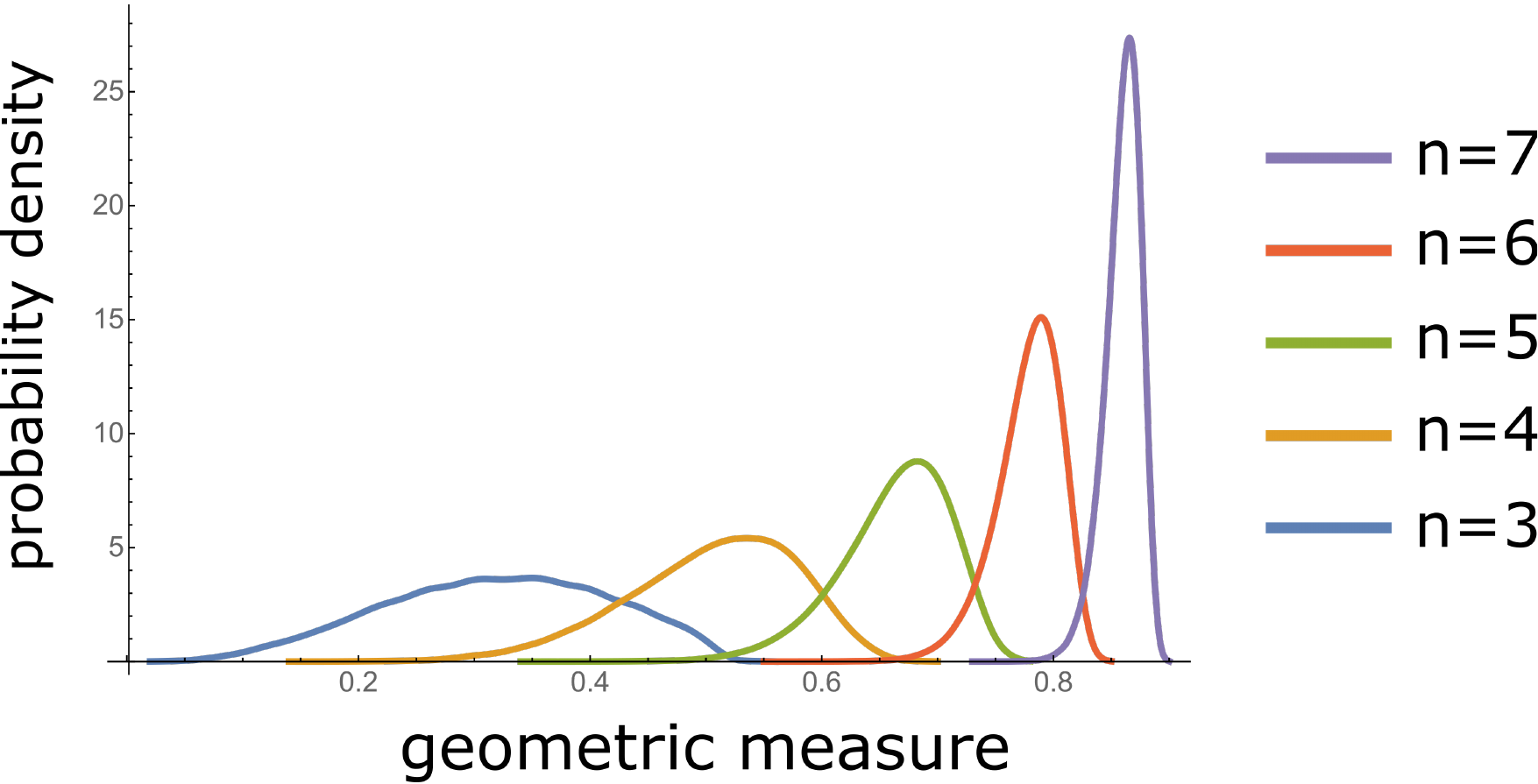}
    \caption{The probability distribution of the geometric measure for different multi-qubit systems. It can be easily seen that the maximum of the distribution is shifted to higher values when the number of parties increases. From the approximation to the distribution, we can also calculate the first and second moments, see Tab.~\ref{tab:expected_entanglement}.
    \label{fig:typical_entanglement}}
  %  \label{fig:app_startpoints}
\end{figure}

In the following we numerically approximate the distribution of entanglement in multi-qubit systems with respect to the Haar measure. The distribution is presented in Fig.~\ref{fig:typical_entanglement}. For this we have randomly sampled $n$-qubit states and computed the corresponding geometric measure. For the cases $n=3,4$ we have chosen $10^5$ states. To compute the geometric measure for each of those states, $20$ iterations and $20$ random starting points were sufficient in order to get a robust output. For $n=5,6,7$ we have chosen $10^6$ states and the calculation of the geometric measure was implemented using $80$ iterations and $100$ random starting points, yielding a robust estimate for the geometric measure. More precisely, we regard a computation of the geometric measure as robust, if the variance of the different outputs is smaller than $10^{-7}$. It should be noted that there is no analytical expression either for the exact distribution of $G$ w.r.t. the Haar measure nor for its moments, e.g., expectation value or variance. However, using the geometric measure of the sampled quantum states we can obtain estimates for the moments of the distribution. In order to get optimal estimates for higher moments, i.e., with small error probabilities, one can use the method of $U$-statistics~\cite{u_statistics_hoeffding_1992}. For the mean value, Hoeffding's inequality yields that the estimated expectation value $\mathbb{E}[G]$ is very close to the true value with very high success probability. The exact numerical values can be found in Tab.~\ref{tab:expected_entanglement}. The geometric measure of typical tensors was also recently numerically investigated in 
Ref.~\cite{typical_entanglement_fitter_2022}.

More generally, there exist bounds which limit the measure of sets of states having a small amount of entanglement~\cite{Bremner2009}, see also Appendix \ref{app:upper_bounds}. Particularly, for multi-qubit system composed of more than $11$ constituents, one has
\begin{align} \label{eq:bound_low_entanglement}
    \mu_{Haar} [\lbrace \vert \varphi \rangle \in (\mathbb{C}^{2})^{\otimes n} \, : \, G(\vert \varphi \rangle) < 1 - \frac{8n^{2} }{2^{n}} \rbrace] \leq e^{-n^{2}}.
\end{align}
Therefore, if $n$ is large enough, almost all states will have a geometric measure close to $1$.
\begin{table}
\centering
\begin{tabular}[t]{l|c|c}
\hline
\hline
System & $\mathbb{E}[G]$ & $\text{Var}[G]$\\
\hline
$3$ qubits & $0.3089$ & $0.0094$\\
$4$ qubits & $0.4950$ & $0.0054$\\
$5$ qubits & $0.6534$ & $0.0023$\\
$6$ qubits & $0.7731$ & $0.0008$\\
$7$ qubits & $0.8570$ & $0.0002$\\
\hline
\hline
\end{tabular}
\caption{The first two moments of the geometric measure $G$ for small multi-qubit systems. The expected amount of entanglement increases with the number of parties and concentrates around its means, which is indicated by the decrease of the variance. \label{tab:expected_entanglement}}
\end{table}%

Further, we have analyzed how the algorithm from the main text behaves for different starting points for multi-qubit systems. For five different Haar random starting points we have chosen $300$ iterations with the same step size $\epsilon=0.05$. Within the computation of the geometric measure in each step, we used $50$ iterations and $50$ random product states. The exact behaviours can be found in Fig.~\ref{fig:randomstarts}. Here it is important to note that after the iterations each of the states has approximately the same geometric measure. This indicates that at least for the case of a small number of parties, the algorithm is stable with respect to local optima. 
\begin{figure}[t] 
    \centering
    \includegraphics[width=0.9\columnwidth]{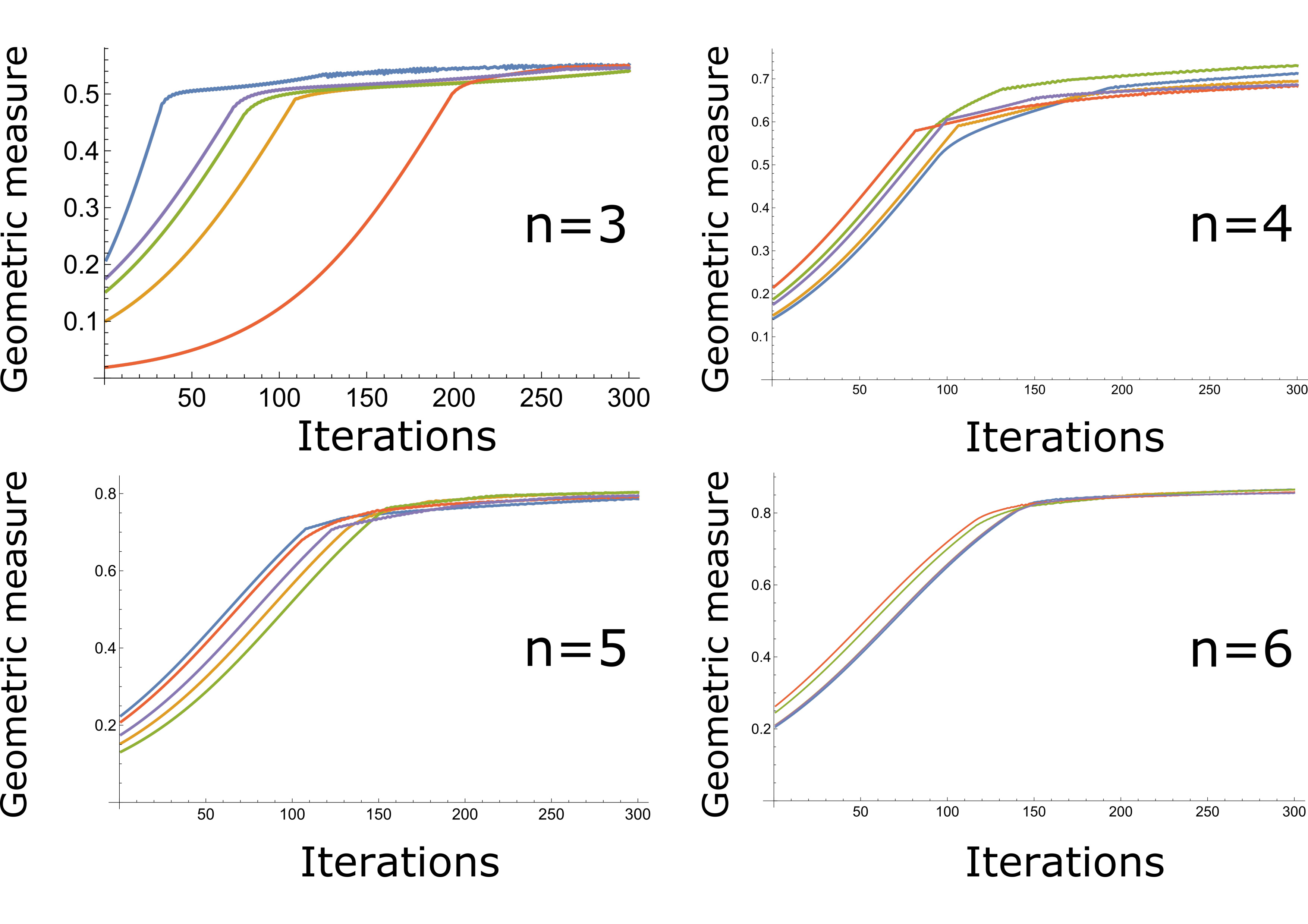}
    \caption{Trajectory of the geometric measure of the iterates of the algorithm for  different staring points and different numbers of qubits.
    While for the case $n=3$ the different trajectories differ, this deviation becomes smaller if the number of parties increase. However, it should be noted that after $300$ iterations each state displays roughly the same amount of entanglement, showing that at least for the case of $n=3,4,5,6$ the algorithm is robust w.r.t. to local optima. It can be seen that for increasing $n$ also the geometric measure of the initial state increases, as the expected geometric measure increases. Further, the different trajectories are getting more narrow with increasing $n$, reflecting the concentration property of the geometric measure.
    \label{fig:randomstarts}}
\end{figure}

%%%%%%%%%%%%%%%%%%%%%%%%%%%%%%%%%%%%%%%%%%%%%%%%%%%%%%%%%%%%%
%%%%%%%%%%%%%%%%%%%%%%%%%%%%%%%%%%
%%%%%%%%%%%%%%%%%%%%%%%%%%%%%%%%%%%%%%%%%%%%%%%%%%%%%%%%%%%%%

\section{Known AME graph states and their graphs}
\label{app:ame_and_graphs}
In this section we present the graphs of the corresponding graph states mentioned 
in the main text and explain their construction. There are two equivalent methods
how a graph state could be defined, namely via a quantum circuit in terms of a 
sequence of commuting unitary two-qubit operations or alternatively using 
the stabilizer formalism~\cite{graph_states_hein_2004}. Here we use the 
first approach. A graph $G=(V,E)$ is a set of $\vert V \vert$ vertices and 
some edges $E \subset V\times V$ connecting them. The graph state $\vert G \rangle $ 
associated to a graph $G$ is a pure quantum state on $(\mathbb{C}^{2})^{\otimes \vert V \vert}$ 
built up of Ising type interactions according to
\begin{align}\label{eq:graphstate}
    \vert G \rangle = \!\! \prod_{(a,b) \in E} \textsf{CZ} _{a,b} \!\! \vert + \rangle^{\otimes \vert V \vert},
    \;\text{with}\;
    \textsf{CZ} _{a,b} = \sum_{k=0}^{1} \vert k \rangle \langle k \vert_{a} \otimes Z_{b}^{k},
\end{align}
where $Z_{b}$ denotes $\sigma_{z}$ only acting on system $b$ and similar 
for $\openone_{a}$. The operator $\textsf{CZ}_{a,b}$ denotes the controlled $\sigma_{z}$ on system $b$ conditioned to system $a$. 
In terms of the stabilizer formalism, the graph state $\vert G \rangle$ corresponds to the unique common eigenstate of all stabilizing operators with eigenvalue $+1$. More generally, given a graph state $\vert G \rangle$ one can generate a orthonormal basis for $(\mathbb{C}^{2})^{\otimes \vert V \vert }$ of common eigenstates, labelled by their eigenvalue signature. For given $\vert G \rangle$ we can obtain this basis by $\vert \omega \rangle = \sigma_{z}^{\omega} \vert G \rangle$ where $\omega = (\omega_{1},...,\omega_{\vert V \vert}) \in \lbrace 0,1 \rbrace^{\vert V \vert}$. The $\text{AME}(4,4)$ state can be built up out of $8$ 
qubits which are grouped together as $A= \lbrace 1,2 \rbrace, B= \lbrace 3,4 \rbrace,C= \lbrace 5,6 \rbrace,D= \lbrace 7,8 \rbrace$, while the Ising type interaction is implemented 
as described in Eq.~\eqref{eq:graphstate}. 
\begin{figure*}[t]
    \centering
    \begin{minipage}{.3\textwidth}
        \centering
        \includegraphics[width=0.75\linewidth]{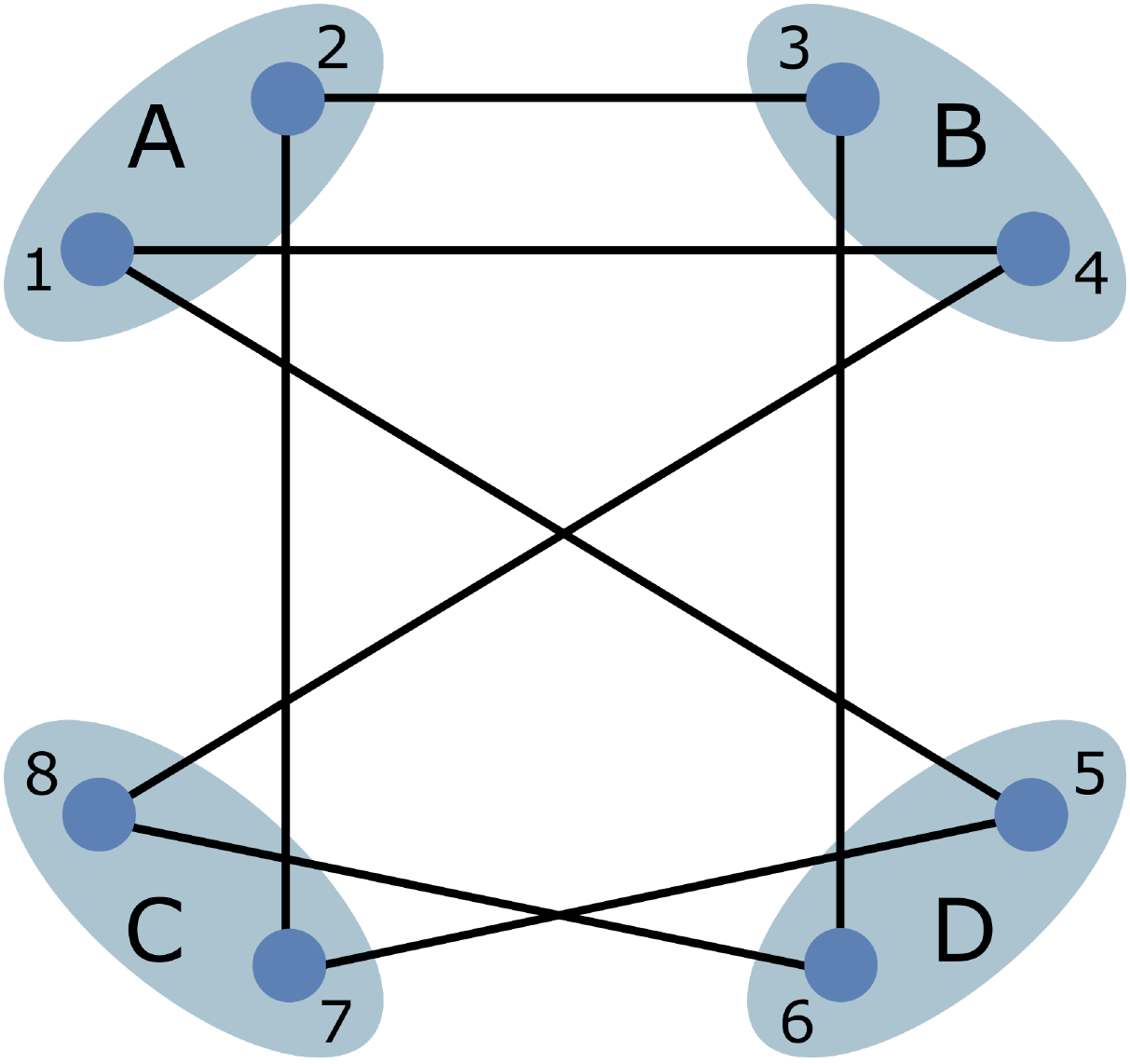}
        \caption{Graph of the known $\text{AME}(4,4)$ state~\cite{AME_qudit_helwig_2013}.}
        \label{fig:prob1_6_2}
    \end{minipage}%
     \begin{minipage}{.05\textwidth}
     \mbox{ }
      \end{minipage}
    \begin{minipage}{0.64\textwidth}
        \centering
        \includegraphics[width=\linewidth]{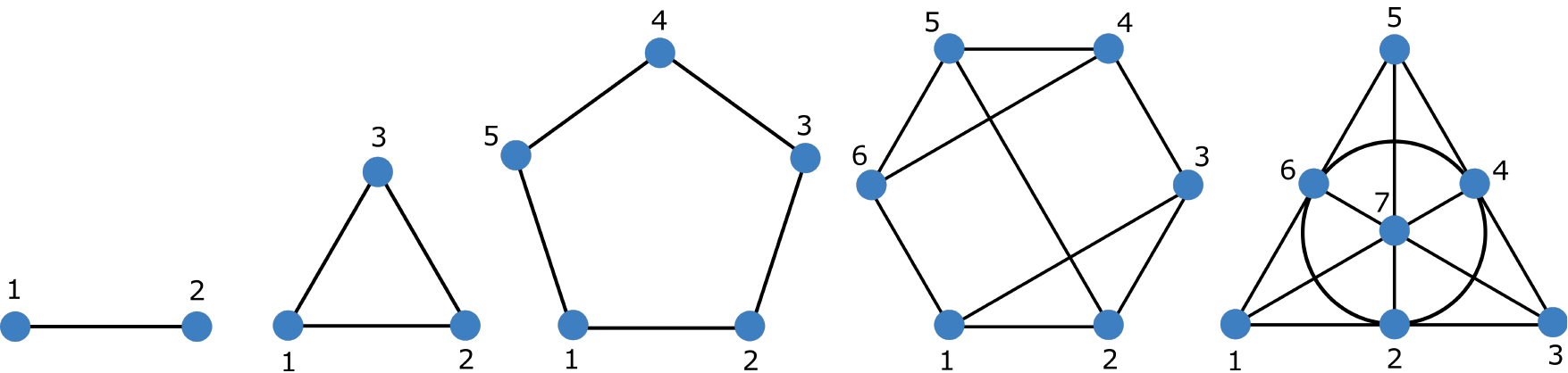}
        \caption{Highly entangled qubit graph states. The first four graphs yielding the known AME states on two/three/five/six qubits respectively. The last graph state, the Fano graph state~\cite{huber2017}, is a $2$-uniform state on seven qubits  where $32$ of its $35$ three-body marginals are maximally mixed.}
        \label{fig:prob1_6_1}
    \end{minipage}
\end{figure*}

%%%%%%%%%%%%%%%%%%%%%%%%%%%%%%%%%%%%%%%%%%%%%%%%%%%%%%%%%%%%%
%%%%%%%%%%%%%%%%%%%%%%%%%%%%%%%%%%
%%%%%%%%%%%%%%%%%%%%%%%%%%%%%%%%%%%%%%%%%%%%%%%%%%%%%%%%%%%%%

\section{Approximation to states with fixed stabilizer rank}\label{app:stabilizer}
In this section we explain how the algorithm can be used to find states which can not be approximated well by states of a fixed stabilizer rank. For this purpose, we will first introduce the concept of stabilizer rank and explain why it is an important tool for the classical simulation of quantum computations. Afterwards, we introduce the modified algorithm and present numerical results for a slight variant of the stabilizer rank which is numerically much more tractable. 

\subsection{The stabilizer rank}
The stabilizer rank $\chi$ of a multi-qubit system was first introduced in the context of simulation of Pauli-based quantum computation (PBC), a specific computational model~\cite{stabilizer_overlap_bravyi_2016}. A PBC on $n$ qubits is given by a sequence of time steps $t=1,...,n$, where for each $t$ one performs a nondestructive eigenvalue measurement of some Pauli operator $P_{t} \in \lbrace \openone, X,Y,Z \rbrace^{\otimes n}$ and the choice of the particular Pauli operator $P_{t+1}$ can depend on the outcome of the measurement at time step $t$. Each of the $n$ qubits is initialized in a so called magic state $ \vert H \rangle = \cos(\frac{\pi}{8}) \vert 0 \rangle + \sin(\frac{\pi}{8}) \vert 1 \rangle$. After all $n$ operators $P_{1},...,P_{n}$ have been measured, yielding an outcome string $s=(s_{1},...,s_{n}) \in \lbrace \pm 1 \rbrace^{n}$, the final state is discarded. Finally, the obtained outcome string $s$ is converted via a classical post-processing into an output $b \in \lbrace 0,1 \rbrace$.

While the main aim of Ref.~\cite{stabilizer_overlap_bravyi_2016} is to quantify how many classical resources (time) are needed in order to simulate a PBC on $(n+k)$ qubits by a PBC using only $n$ qubits, their introduced method turned out to also improve the cost of a brute force full classical simulation. It can be shown that a brute force classical simulation of a $n$ qubit PBC, which includes state operations in the computational basis, has a simulation cost (time) of $\mathcal{O}(n2^{n})$, where $n$ denotes the number of qubits involved in the computation. However, one can devise an algorithm which classically simulates a $n$ qubit PBC  in only $2^{\alpha n} \text{poly}(n)$, where $\alpha \approx 0.94$~\cite{stabilizer_overlap_bravyi_2016}. The idea of the proof is to expand certain $n$-qubit states as a linear combination of stabilizer states, as $n$-qubit stabilizer states only require $\mathcal{O}(n^{2})$ classical bits to store~\cite{gottesman_thm_1998}. The minimal number of stabilizer states needed is then called the stabilizer rank.

More precisely, the set of stabilizer states consists of all states $\vert s \rangle$ that can be generated from some $n$-qubit Clifford operation $U$ via $\vert s \rangle = U \vert 0 \rangle^{n}$, where a Clifford operation is a unitary operator which maps via a conjugate action tensor products of Pauli operators to tensor products of Pauli operators, that is,
\begin{align}
    \text{Cl}_{n} := \lbrace V \in \mathcal{U}(2^{n}) \, \vert \, V^{\dagger} P V  \in \mathcal{P}_{n} \, \forall P \in \mathcal{P}_{n}  \rbrace .
\end{align}
Equivalently, one can define $\text{Cl}_{n}$ as the group generated by the gates $\textsf{H},\textsf{S},\textsf{CNOT}$. 
For an arbitrary $n$-qubit state $\vert \psi \rangle$, its stabilizer rank $\chi$ is defined as the minimum number $r \geq 1$ of stabilizer states needed to represent this state, that is, 
\begin{align}
    \vert \varphi \rangle = \sum_{k=1}^{r} \alpha_{k} \vert s_{k} \rangle.
\end{align}

\subsection{The modified algorithm and results}
So far, only for very special states the stabilizer rank is known and only rude upper bounds on the  maximal attainable $\chi$ for a $n$-qubit system have been derived. It is therefore highly desirable, from the practical viewpoint of classical simulation as well as from a theoretical perspective, to know which states cannot be approximated well by stabilizer rank $\chi$ states. For this, we define the stabilizer measure $\mathcal{S}_{k} : (\mathbb{C}^{2})^{\otimes n} \rightarrow \mathbb{R}$ with
\begin{align} \label{eq:stab_measure}
    \vert \psi \rangle \mapsto \mathcal{S}_{k} (\vert \psi \rangle) := 1 - \text{sup} \lbrace \,  \vert \langle \varphi \vert \omega \rangle \vert^{2} \, : \, \chi(\vert \omega \rangle )  =  k \rbrace .
\end{align}
For a fixed $k$, the stabilizer measure in Eq.~\eqref{eq:stab_measure} directly allows for an application of the algorithm. First, generate the set of all stabilizer states and form all subsets of size $k$. Second, draw an initial state $\vert \psi \rangle$ at random. Then, compute the closest stabilizer rank $k$ state, i.e., for given set of stabilizers $\vert s_{1} \rangle,..., \vert s_{k} \rangle$, one has to maximize the overlap of the given state $\vert \psi \rangle$ with $\vert \omega \rangle = \sum_{j=1}^{k} \alpha_{j} \vert s_{j} \rangle$. This is a constraint optimization over $k$ complex parameters $\alpha_{1},...,\alpha_{k}$ subject to the normalization of the stabilizer rank $k$ state. If the best approximation $\vert \omega \rangle$ is found, one updates $\vert \psi \rangle \mapsto (1/\mathcal{N})(\vert \psi \rangle + \theta \vert \eta \rangle)$ where $\vert \eta \rangle = (\openone - \vert \omega \rangle \langle \omega \vert) \vert \psi \rangle$. This procedure is then iterated.

\begin{figure}[b] 
    \centering
    \includegraphics[width=0.95\columnwidth]{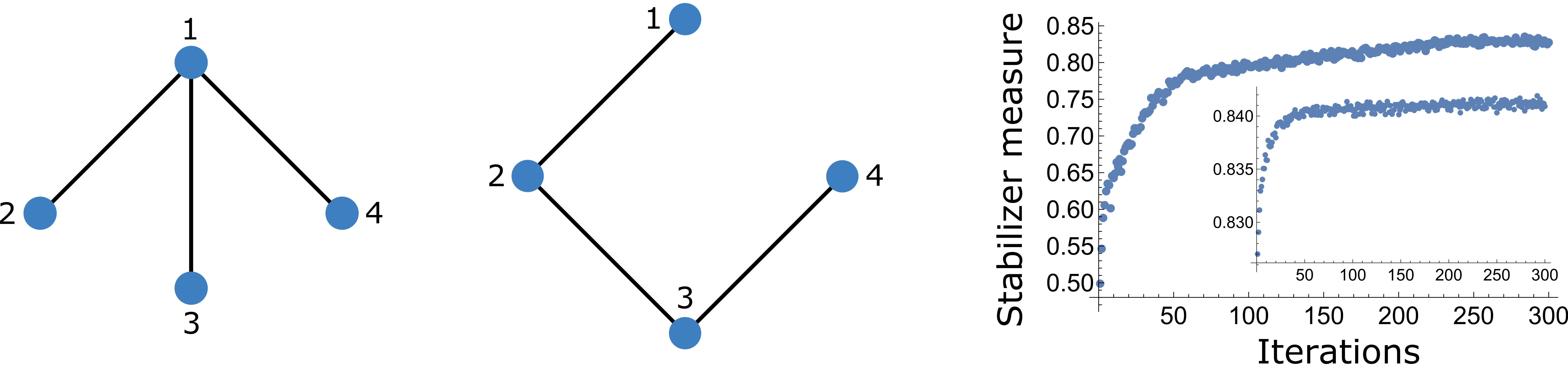}
    \caption{The graph corresponding to the four particle GHZ state (left) and the graph of the four particle cluster state (middle). Performance of the algorithm for the case of four particles (right). The optimization is implemented for stabilizer rank $\chi=2$ in the restricted model of $32$ stabilizer states, corresponding to GHZ and cluster graph state basis. We start with a random initial state and a step size of $\epsilon=0.1$ (outset). Then, we change step size to $\epsilon = 0.01$ (inset).}
    \label{fig:stabilizer}
\end{figure}

However, the number of $n$ qubit stabilizers $S(n)$ can be computed explicitly~\cite{number_of_qubit_stabilizers_gross_2006} and is given by $S(n) = 2^n \prod_{j=1}^{n} (2^{j}+1)$, which scales exponentially in the number of qubits. For instance, one has $S(3)= 1080$, $S(4) = 36720$ and $S(5) = 2423520$. Further, in order to evaluate the overlap with stabilizer rank $k$ states, one has to inspect all $\binom{S(n)}{k}$ different combination and optimize over the $k^{2}-1$ real parameter which becomes infeasible even for small $n,k$. We therefore consider a non-trivial toy model, where we restrict to a specific subset of stabilizer states. In particular, we consider the set consisting of the graph state bases  corresponding to four qubit GHZ-graph and the four-qubit cluster-graph, see Fig.~\ref{fig:stabilizer}, as for four qubits there are only two inequivalent graphs. For each graph, the set of all eigenstates with a different eigenvalue signature is a basis for $(\mathbb{C}^{2})^{\otimes 4}$ and thus our toy model comprises $32$ stabilizer states. Note that for one graph, those states form an orthonormal basis, while the set of eigenstates of two graphs displays a highly nontrivial geometric structure. If the optimization in Eq.~\ref{eq:stab_measure} only runs over this restricted set, we write $\mathcal{S}_{k}^{R}$ for the corresponding stabilizer measure. Note that in difference to the geometric measures $G_{m}$, two LU-equivalent states do not need to have the same measure $\mathcal{S}_{k}$, as we try to maximize the overlap w.r.t. a set that is not invariant under LU operations. Of course, the measure $\mathcal{S}_{k}$ of two states which are Clifford equivalent is the same.

We have implemented the algorithm for the restricted toy model, see also Fig.~\ref{fig:stabilizer} (right). We initialized with a random state and iterated $300$ times with a step size of $\theta = 0.1$ until $\mathcal{S}_{k}^{R}$ does not increase anymore. Then we change the step size to $\theta = 0.01$ and iterate again $300$ times. 
The algorithm converges to a state $\vert \psi_{\text{max}} \rangle$ yielding $\mathcal{S}_{2}^{R}(\vert \psi_{\text{max}} \rangle) \approx 0.8415$. A comparison of $\mathcal{S}_{2}^{R}$ for different four-qubit states can be found in Table~\ref{stabilizertab}.

\begin{table}[b]
\centering
\begin{tabular}[t]{lccccc}
\hline
\hline
State $ \vert \psi \rangle$ & $\vert M \rangle$ & $\vert L \rangle$ & $\vert H \rangle^{\otimes 4}$ & $\vert R \rangle^{\otimes 4}$ & $\vert \psi_{\text{max}} \rangle$ \\
\hline
Measure $\mathcal{S}_{2}^{R}$ & $1/3 \approx 0.333$ & $1/2$ & $\approx 0.4457$ & $0.4604$ &$0.8415$ \\
\hline
\hline
\end{tabular}
\caption{Results of the optimization with respect to the restricted stabilizer measure $\mathcal{S}_{2}^{R}$ for rank two. The states $\vert H \rangle$ and $\vert R \rangle$ are magic states for which it is conjectured to have the smallest possible stabilizer rank among all non stabilizer single-qubit states~\cite{stabilizer_overlap_bravyi_2016}. The state $\vert \psi_{\text{max}} \rangle$ refers to the state found by the algorithm for which $\mathcal{S}_{2}^{R}$ appears to be maximal.}
\label{stabilizertab}
\end{table}

As already mentioned, except from an asymptotic analysis for special states, not much is known about the minimal overlap which a quantum state can have with the set of all stabilizer states~\cite{upper_bounds_stabrank_gosset_2021}. For instance, for tensor products of the magic states $\vert H \rangle $ or $\vert R \rangle = \cos(\beta) \vert 0 \rangle + e^{i \pi /4} \sin(\beta) \vert 1 \rangle$ with $\beta = (1/2)\arccos(1/\sqrt{3})$, the maximum overlap with all stabilizer states $\vert \varphi \rangle$ on $n$ qubits can be bounded by $\vert \langle \varphi \vert H^{\otimes n} \rangle \vert \leq 2^{-\Omega(n)}$~\cite{stabilizer_overlap_bravyi_2016}. However, the exact constants appearing in $\Omega(n)$ remain hidden. Here our algorithm could yield states which are difficult to approximate with a fixed stabilizer rank, thus giving new insights into the structure of stabilizer states which may be helpful to find improved bounds.

%%%%%%%%%%%%%%%%%%%%%%%%%%%%%%%%%%%%%%%%%%%%%%%%%%%%%%%%%%%%%
%%%%%%%%%%%%%%%%%%%%%%%%%%%%%%%%%%
%%%%%%%%%%%%%%%%%%%%%%%%%%%%%%%%%%%%%%%%%%%%%%%%%%%%%%%%%%%%%

\section{Approximation to states with a fixed Schmidt-rank or a fixed bond dimension}\label{app:schmidt_rank}
In this section we explain how the algorithm can be used to find states which can not be well approximated by states with a fixed Schmidt-rank or, for the particular case of matrix product states, with a fixed bond dimension. We first recapitulate the notion of Schmidt rank, tensor rank and border rank. We then proceed by introducing a slight variant of the seesaw algorithm from the main text for computing the best product state approximation, which allows for a computation of the best rank-$k$ approximation. We proceed by applying the modified algorithm to a four qubit system where we approximate with Schmidt rank $2$ and $3$. We conclude this section by discussing a further application, where the measure of interest quantifies how well a given state can be approximated by matrix product states of given bond dimension. 

\subsection{Schmidt rank, tensor rank and border rank}
Apart from the geometric measure there are also other quantifier of entanglement offering different operational interpretations. 
One such possibility is to quantify the amount of entanglement in a system by its \textit{dimensionality}. High-dimensional quantum entanglement can be exploited to tolerate larger amounts of noise in quantum communication protocols, making them a valuable resource~\cite{high_dim_entanglement_review_zeilinger_2020}.

Let $\vert \varphi \rangle$ be an $n$-particle quantum state where each constituent is a $d$-level quantum system. Any such  $\vert \varphi \rangle $ can be written as
\begin{align} \label{eq:schmidt_rank}
    \vert \varphi \rangle = \sum_{k=1}^{R} \mu_{k} \vert \pi_{k}^{(1)} \rangle \otimes \cdots \otimes \vert \pi_{k}^{(n)} \rangle,
\end{align}
where $\vert \pi_{k}^{(j)} \rangle \in \mathbb{C}^{d}$ and $\mu_{k} \in \mathbb{C}$. Note that the different summands $\vert \pi_{k}^{(1)} \rangle \otimes \cdots \otimes \vert \pi_{k}^{(n)} \rangle$ do not have to be orthogonal as in the case of the Schmidt decomposition of bipartite pure states. 
The  minimal number of terms $R$ needed in order to decompose $\vert \varphi \rangle$ into the form in Eq.~\eqref{eq:schmidt_rank}, i.e., into a sum of rank-$1$ tensors, is called the tensor rank of $\vert \varphi \rangle$ and is denoted by $\rk(\vert \varphi \rangle)$. The decomposition of $\vert \varphi \rangle$ in Eq.~\eqref{eq:schmidt_rank} is also called rank decomposition or minimal CP decomposition~\cite{tensor_decompositions_review_kolda_2009}.

The Schmidt measure $P$ of a state $\vert \psi \rangle$ is defined as $P( \vert \varphi \rangle) := \log_{2} \rk(\vert \varphi \rangle)$~\cite{schmidt_measure_eisert_2001}. Clearly, in the case of a bipartite system, the minimal number of rank-$1$ terms needed is given by the Schmidt rank of the state. Once $P$ is defined for pure states, one can extend it to the full state space of mixed quantum states in a natural way, i.e., via a convex roof construction~\cite{convex_roof_uhlmann_2000}. The Schmidt measure gives also rise to an approximation version similar to the geometric measure. For this, define the generalized geometric measure $G_{k}: (\mathbb{C}^{d})^{\otimes n} \rightarrow \mathbb{R}$  with
\begin{align}
    \vert \varphi \rangle \mapsto G_{k} (\vert \varphi \rangle) := 1- \text{sup} \lbrace \,  \vert \langle \varphi \vert \omega \rangle \vert^{2} \, : \, \rk (\vert \omega \rangle ) = k \rbrace,
\end{align}
which measures how well a given quantum state $\vert \varphi \rangle$ can be approximated by a state of tensor rank $k$.

Therefore, it gives more detailed information about the entanglement structure present in the state $\vert \varphi \rangle$, and reproduces the geometric measure $G$ for the special case $k=1$. Clearly, the generalized geometric measure is a decreasing function in the rank $k$, that is, for given $\vert \varphi \rangle$ we have $G (\vert \varphi \rangle)=G_{1} (\vert \varphi \rangle) \leq G (\vert \varphi \rangle) \leq \cdots \leq G_{\mu}(\vert \varphi \rangle)$, where $\mu$ is the maximal rank that a tensor in $(\mathbb{C}^{d})^{\otimes n}$ can have. For instance, for the bipartite case $\mathbb{C}^{d_{A}} \otimes \mathbb{C}^{d_{B}}$, the Schmidt decomposition directly leads to $\mu = \text{min} \lbrace d_{A},d_{B} \rbrace$. For the case of three qubits, it is known~\cite{rankbounds_jaja_1979} that $\mu =3$, but for the general case only not tight upper bounds exist. Obviously, if we find $G_{k}(\vert \varphi \rangle ) =0$, then $G_{\tilde{k}}(\vert \varphi \rangle) =0$ for all $\tilde{k} \geq k$. The algorithm for maximizing the measure $G_{k}$ then proceeds as follows. First, draw the initial state $\vert \varphi \rangle$ at random and compute its best rank-$k$ approximation, that is, 
\begin{align}
    \vert \omega \rangle := \text{argmin} \, \lbrace \vert \omega \rangle \, : \, 1 - \vert \langle \varphi \vert \omega \rangle \vert^{2} \, : \, \rk(\vert \omega \rangle) =k \rbrace 
\end{align}

\subsection{Computing rank-$k$ approximations}

There are different methods to compute low rank approximations to a given tensor, most prominently the CP-ALS algorithm, which updates each mode individually by a least-squares optimization while keeping all other modes fixed. For a concise summary of those methods, see Ref.~\cite{tensor_decompositions_review_kolda_2009}. However, to keep notation simple, we will generalize the algorithm mentioned in the main text for higher rank approximations, which we will illustrate for the three-particle case. For a given state $\vert \psi \rangle \in (\mathbb{C}^{d})^{\otimes 3}$ consider the Schmidt decomposition $\vert \psi \rangle_{A \vert BC} = \sum_{k=1}^{D} \lambda_{k} \vert \psi_{k}^{A} \rangle \vert \psi_{k}^{BC} \rangle$. For the overlap with a rank $k$ state given by Eq.~\eqref{eq:schmidt_rank} one finds
\begin{align}
    \langle \psi \vert_{A\vert BC} \pi \rangle = &
    \sum_{j=1}^{D} \sum_{k=1}^{r} \lambda_{j} \alpha_{k} \langle \psi_{j}^{A} \vert \pi^{(A)}_{k} \rangle \langle \psi_{j}^{AB} \vert  \pi^{(B)}_{k} \rangle \vert  \pi^{(C)}_{k} \rangle 
    \nonumber
    \\
     = & \sum_{k=1}^{r}  \langle \tilde{\gamma}_{k} \vert \pi_{k}^{(A)} \rangle,
\end{align}
where we have defined $\langle \tilde{\gamma}_{k} \vert = \sum_{j=1}^{D} \lambda_{j} \alpha_{k} \langle \psi_{j}^{AB} \vert \pi^{(B)}_{k} \rangle \vert  \pi^{(C)}_{k} \rangle \langle \psi_{j}^{A} \vert $. Note that $ \vert \tilde{\gamma}_{k} \rangle$ is not normalized and we denote by $\vert \gamma_{k} \rangle$ its normalized version. The norm is given by $\alpha'_{k} = \sqrt{\langle \tilde{\gamma}_{k} \vert \tilde{\gamma}_{k} \rangle}$. We now update the states of system $A$ as well as the coefficients $\alpha_{k}$ that appear in the decomposition as  $\vert \pi_{k}^{A} \rangle \mapsto \vert \gamma_{k} \rangle$ and $\alpha_{k} \mapsto \alpha'_{k}$. This procedure is iterated with respect to each mode and repeated many times. In order to reduce the sensitivity for local maxima, different randomly chosen starting points $\vert \pi \rangle$ should be used.

To test the reliability of the algorithm, we apply it to quantum states for which the tensor rank is already known. For instance, one has $\rk(\vert \text{GHZ} \rangle) =2$, $\rk(\vert \text{W} \rangle) =3$ and $\rk(\vert \text{W} \rangle^{\otimes 2}) =7$~\cite{tensor_rank_2W_chitambar_2010}. 
However, one should notice that for tensor approximations of higher rank ($k \geq 2$) the problem of degeneracy exists, reflecting that the set of tensors of rank $k$ is open which is in contrast to the matrix case. In particular, for the bipartite (matrix) case the Eckart-Young theorem ensures that the best rank $k$ approximation is simply given by the eigenvectors corresponding to the $k$ largest singular values. Further it is known that deflation techniques, i.e., compute best rank $1$ approximation, subtract it and then iterate, are not fruitful for tensors.

A tensor is degenerate if it may be approximated \textit{arbitrarily well} by tensors of lower rank. Also well known~\cite{hardness_approximation_2008,tensor_decompositions_review_kolda_2009}, we will discuss this phenomenon for the case of the W state. The W state $\vert \text{W} \rangle = (1/\sqrt{3})( \vert 001 \rangle + \vert 010 \rangle + \vert 100 \rangle)$ has tensor rank $3$, but can be approximated arbitrarily well by tensors of tensor rank $2$. Consider the family of rank $2$ states
\begin{align}
    \vert \pi(\alpha) \rangle = & \frac{\alpha}{N_{\alpha}} (\vert 0 \rangle + \alpha^{-1} \vert 1 \rangle)^{\otimes 3} - \frac{\alpha}{N_{\alpha}} \vert 0 \rangle^{\otimes 3} 
    \nonumber \\ &\text{with} \, \, N_{\alpha} = \sqrt{3 + 3 \alpha^{-2} + \alpha^{-4}},
\end{align}
where $N_{\alpha}$ assures that $\langle \pi(\alpha) \vert \pi (\alpha) \rangle =1$ for all $\alpha >0$. It directly follows that $\langle \text{W} \vert \pi(\alpha) \rangle = \sqrt{3}/\sqrt{3 + 3x^{-2} + x^{-4}}$ which tends to $1$ for $x \rightarrow \infty$. To emphasize that problem, one also says that the W state has border rank $2$. Therefore, it is not surprising that our approximation algorithm yields that $G_{2}(\vert \text{W} \rangle)$ is numerically $0$. 
As it turns out, the generalized geometric measure $G_{k}$ has also the advantage that it distinguishes more clearly between different forms of entanglement, e.g., the L state and the M state. In particular, one finds that $G_{3}(\vert \text{L} \rangle) =0$ while $G_{3}(\vert \text{M} \rangle) \approx 0.2626$. 

\subsection{The modified algorithm and results}

The smallest case where a maximization of $G_{2}$ could be considered is in principle $(\mathbb{C}^{2})^{\otimes 3}$. Here it is known that there are only two classes of genuine multipartite entanglement with respect to SLOCC, namely GHZ and W~\cite{three_qubits_ineuivalent_entangled_duer_2000}. Further, the tensor rank is monotonically decreasing under SLOCC operations~\cite{schmidt_measure_eisert_2001}. In particular, one can show that the set of all tensors of rank three is the closure of the SLOCC orbit of the W state. However, the W state has border rank two and thus all states within his orbit. This implies the nonexistence of border rank three tensors in  $(\mathbb{C}^{2})^{\otimes 3}$. As a consequence, the algorithm cannot be applied to this case, as in each step the given state can be approximated arbitrarily well by states of rank $2$.
Therefore the first nontrivial case is a system of four qubits with respect to rank two and three. Here, states of border rank three and four exist and are subsets of non-vanishing measure. We implement the algorithm with $20$ different random initial points for each case. For the case of rank $2$, we identify the state to be in $7$ of the $20$ runs the M state and in the other cases the four qubit cluster state. Further, the numerical computation of $G_{2}$ for $\vert C_{4} \rangle$ is stable with respect to a randomization of the starting point and one finds $G_{2}(\vert C_{4} \rangle) = 1/2$.

\begin{figure}[t]
    \centering
    \includegraphics[width=0.9\columnwidth]{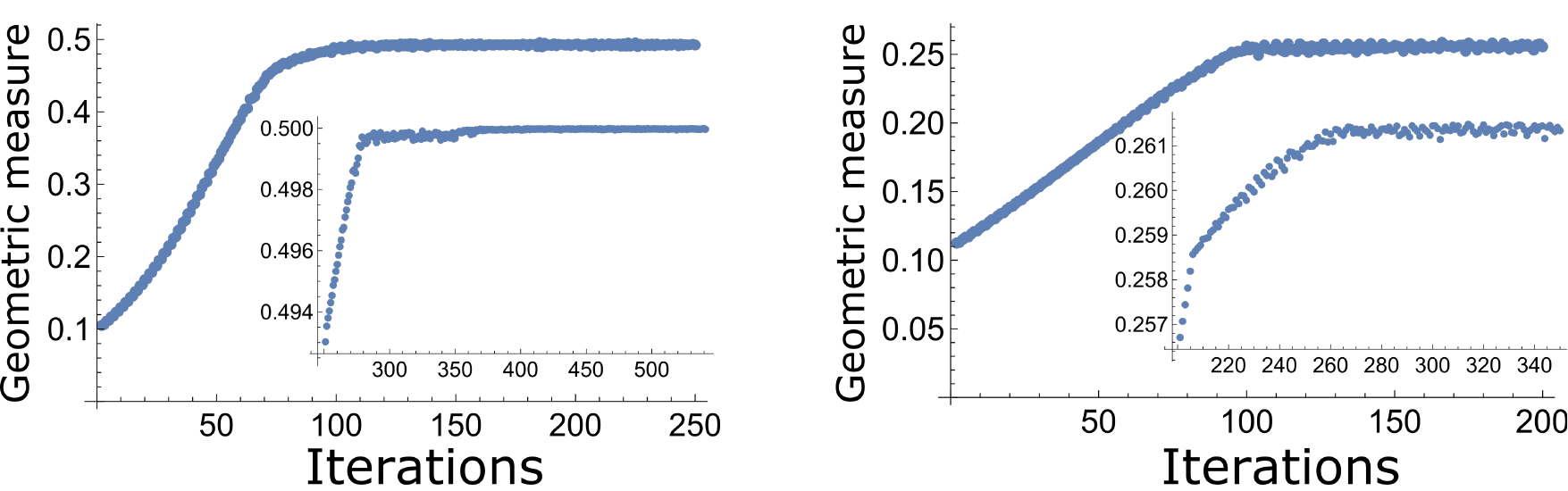}
    \caption{Performance of the algorithm for the case of four qubits and different generalized geometric measures. The left figure shows the maximization of the distance to the set of tensor rank $2$ states, that is, the measure $G_{2}$. We start with a random state and choose a step size of $\epsilon = 0.05$ (main). We compute the rank $2$ best approximation with $100$ iterations and $100$ random initial points. After $150$ iterations, there is no increase in the measure and the values start to fluctuate. We then change the step size to $\epsilon = 0.001$ and compute the best rank $2$ approximation with $200$ iterations and $200$ random initial points (inset). The right figure shows the maximization of the distance to the set of tensor rank $3$ states, that is, the measure $G_{3}$. Similar, we start  with a random state and choose a step size of $\epsilon = 0.05$ (main) and compute the best rank $3$ approximation with $100$ iterations and $100$ random initial points. Then, we change the step size to $\epsilon = 0.001$ and compute the best rank $3$ approximation with $200$ iterations and $200$ random initial points (inset).}
    \label{fig:higher_rank}
\end{figure}

In this case, the rank two approximation found by the algorithm yielding the value $G_{2} (\vert C_{4} \rangle) = 1/2$ can proven to be optimal. The cluster state is a graph state and we can consider the bipartition $23 \vert 14$ (see Fig.~\ref{fig:stabilizer}). This effectively yields the four-dimensional maximally entangled state $\vert 00 \rangle + \vert 11 \rangle +\vert 22 \rangle +\vert 33 \rangle$. Because all singular values coincide, we can choose the best rank two approximation as $\vert 00 \rangle + \vert 11 \rangle$ which yields a squared overlap of $1/2$. As $G_{m} (\vert \psi \rangle_{1234}) \geq G_{m} (\vert \psi \rangle_{12 \vert 34})$ this proves the optimality. 
However, as the M state is not a graph state the situation is different. For the computation we choose $10^{6}$ random starting points and make $300$ iteration for each. The total computation was running for $27$ hours. The optimal approximation found by the algorithm yields $G_{2} (\vert \text{M} \rangle) = 5.000003(9)$.
This suggests that $G_{2} (\vert \text{M} \rangle) = 1/2$ which would imply that the maximizer of $G_{2}$ is not unique anymore in contrast to $G_{1}$. For the case of rank $3$ approximations, the algorithm converges for all $20$ starting points to the M state. With a similar computation as for $G_{2}$ we find that $G_{3}(\vert \text{M} \rangle) = 0.2626050$. However, the algorithm can not be applied to  maximize $G_{k}$ for $ \geq 4$. This is due to the fact, that randomly drawn tensors in $ (\mathbb{C}^{2})^{\otimes 4}$ can be approximated well by rank $4$ tensors.

\subsection{Approximation to matrix product states}
In general, it is an open question which forms of nontrivial quantum dynamics can be simulated efficiently by classical means. However, it is a celebrated result that if the computation only involves pure states containing a restricted amount of entanglement, such an efficient simulation is possible~\cite{matrix_product_states_vidal_2003}. The power of those simulation algorithms rely on the fact that if the entanglement present in the multi-qudit state $\vert \psi \rangle $ is bounded, a low-dimensional sparse representation of $\vert \psi \rangle$ can be derived. This is the so-called matrix product state (MPS) representation.  A MPS of bond dimension $\chi$ can be written as
\begin{align}
    \vert \psi \rangle = \sum_{i_{1},...,i_{n} =1}^{d} \, \text{tr}[A^{[1]i_{1}} \cdot ... \cdot A^{[n]i_{n}}] \, \vert i_{1} \cdots i_{n} \rangle
\end{align}
where $A^{[l] i_{l}} \in \text{Mat}_{\chi}(\mathbb{C})$ for $l=1,...,n$. Any state $\vert \varphi \rangle $ can be represented as a MPS if the bond dimension is sufficiently large~\cite{matrix_product_states_vidal_2003}. For a fixed number of parties $n$ with local dimension $d$ we write $\text{MPS}(\chi)$ for the set of all MPS with bond dimension $\chi$. Notice that $\text{MPS}(\chi)$ is a manifold in state space, which coincides with the set of product states for $\chi =1$ and has a nested structure, i.e., $\text{MPS}(\chi) \subset \text{MPS}(\chi+1)$.
In particular, this shows that computing the best MPS approximation of given bond dimension $\chi$ to a given state $\vert \varphi \rangle$ is a hard problem which includes for $\chi=1$ a NP-hard problem~\cite{hardness_approximation_2008,tensor_problems_hard_2013}.

Recently, a variant of the geometric measure was introduced, where the distance is measured with respect to the set of matrix product states of a given bond dimension $\chi$~\cite{matrix_product_measure_nicokatz_2023}. How well a generic multi-particle quantum state can be approximated by matrix product states of bond dimension $\chi$ is then quantified by $\mathcal{E}_{k} : (\mathbb{C}^{d})^{\otimes n} \rightarrow \mathbb{R}$ with 
\begin{align}
   \vert \psi \rangle \mapsto  \mathcal{E}_{k} (\vert \psi \rangle):=  1- \text{sup} \lbrace \vert \langle \psi \vert \omega \rangle \vert^{2} \, : \, \vert \omega\rangle \in \text{MPS}(\chi) \rbrace
\end{align}

As already pointed out in Ref.~\cite{matrix_product_states_vidal_2003}, the tensor rank is not a continuous function. Consequently, there exist states $\vert \psi \rangle$ that need a high bond dimension to be represented exactly, but very good approximations w.r.t. some norm exist where the optimizer has a significantly lower bond dimension. This property could then be used to obtain efficient simulations with small errors, even for \textit{seemingly} high entangled states. Therefore, finding states which are difficult to approximate by states with fixed bond dimension is important as they do not allow for an efficient classical approximate simulation and can thus be regarded as genuine quantum resources. 

Further, upper bounds on the distance between an arbitrary $n$-particle state $\vert \psi \rangle$ and the set of $\text{MPS}(\chi)$ can be derived~\cite{mps_faithful_verstraete_2006}. Indeed, one can show that for any $\vert \omega \rangle \in \text{MPS}(\chi)$ one has  
\begin{align}\label{eq:mps_upperbound}
    \vert \vert \, \vert \psi \rangle - \vert \omega \rangle \, \vert \vert^{2} \leq 2 \sum_{k=1}^{n-1} \epsilon_{k}(\chi)
\end{align}
with $\epsilon_{k}(\chi) := \sum_{j=k+1}^{d_{k}} \lambda^{[k]}_{j}$, where $\lambda^{[k]}_{j}$ are the Schmidt coefficients of the bipartition $[1...k]\vert [k+1...n]$ and $d_{k}$ is the local dimension of the smaller of the two subsystems for that bipartition. Here our algorithm can find states for which the bound in Eq.~\eqref{eq:mps_upperbound} will be maximal and thus gives information about its tightness.

%%%%%%%%%%%%%%%%%%%%%%%%%%%%%%%%%%%%%%%%%%%%%%%%%%%%%%%%%%%%%
%%%%%%%%%%%%%%%%%%%%%%%%%%%%%%%%%%
%%%%%%%%%%%%%%%%%%%%%%%%%%%%%%%%%%%%%%%%%%%%%%%%%%%%%%%%%%%%%

\section{Approximation to independent triangle preparable states}\label{app:triangle}

\subsection{The triangle scenario}

A class of quantum states with importance for quantum information processing is the set of states which can be prepared in a quantum network~\cite{quantum_internet_kimble_2008,qrepeaters_gisin_2011}. For simplicity, in the following we will restrict to the triangle network and assume that the independent sources distribute pairs of qubit systems. Note that our algorithm can also be applied to more advanced network topologies as well as to higher local dimensions. Following Ref.~\cite{ITN_kraft_2021}, we call this network the independent triangle network, abbreviated ITN. From the structure of the network, see also Fig.~\ref{fig:triangle} (middle), we obtain that a state $\vert \psi \rangle$ can be prepared in the ITN if and only if there exist unitaries $U_{A},U_{B},U_{C}$ and bipartite states $\vert a \rangle, \vert b \rangle, \vert c \rangle$ such that
\begin{align}
    \vert \psi \rangle = U_{A} \otimes U_{B} \otimes U_{C} \vert abc \rangle
\end{align}
Here it is important to note that the order of the subsystems is different for the unitaries and the states. For instance, $U_{A}$ acts on the joint system $A_{1}A_{2}$ while $\vert a \rangle$ defines the joint state between $B_{2}C_{1}$. We denote the set of all states that can be prepared by means of two-qubit sources by $\Delta_{I}$. However, it turns out that $\Delta_{I}$ admits a highly nontrivial structure~\cite{ITN_kraft_2021}. Indeed, $\Delta_{I}$ is not convex, certain separable (product) states are not contained and it is a subset of measure zero within the entire set of quantum states. This renders its analysis and characterisation difficult and not much is known about the structure of this set.

\begin{figure}[t]
    \centering
    \includegraphics[width=0.9\columnwidth]{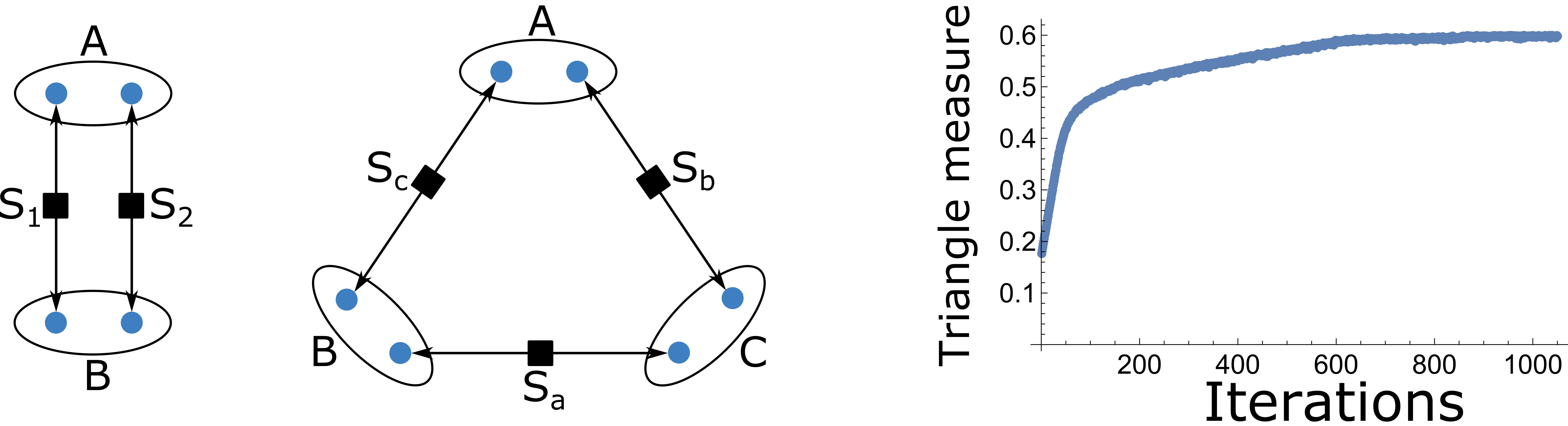}
    \caption{The genuine multilevel scenario (left). Each of the independent sources $S_{1}$ and $S_{2}$ generates a bipartite entangled system and sends one part to A and the other to B. Consequently, each of the parties receives two physical systems on which a joint unitary transformation can be performed. The set of all states that can be obtained in this scenario is abbreviated by $\Delta_{\text{GM}}$. The independent triangle scenario (middle). Three statistically independent sources $S_{a},S_{b}, S_{c}$, distributing physical systems to the parties A,B,C. Upon receiving the two independent systems, each party can perform a joined unitary on the hold particles. The set of all states that can be obtained in this scenario is abbreviated by $\Delta_{I}$. The performance of the algorithm for the maximization of the triangle measure $\mathcal{T}$ (right).   }
    \label{fig:triangle}
\end{figure}

As triangle states constitute an important resource, it is an interesting question how well a given quantum state can be approximated by states from $\Delta_{I}$. In order to quantify this property, we define the triangle measure $\mathcal{T}:  (\mathbb{C}^{d})^{\otimes 6} \rightarrow \mathbb{R}$ with
\begin{align}
    \vert \psi \rangle \mapsto \mathcal{T}(\vert \psi \rangle) := 1 - \text{sup} \, \lbrace \vert \langle \psi \vert \omega \rangle \vert^{2} \, : \, \vert \omega \rangle \in \Delta_{I} \rbrace 
\end{align}
Clearly, for any ITN preparable state we have $\mathcal{T}=0$. 

\subsection{Computing the best network state approximation}

Similar to the problem of finding the best rank one or best rank $k$ approximation to a given state in the computation of the generalized geometric measures $G_{k}$, one can also devise a seesaw type algorithm for computing the ITN state which maximizes the overlap~\cite{ITN_kraft_2021}. We will shortly recapitulate this algorithm. For simplicity, let us assume that the given target state $\vert \psi \rangle$ is composed out of six qubits. We initialize the algorithm with a unitary for each party $U_{A},U_{B},U_{C}$  and a state for each source $\vert a \rangle$,$ \vert b \rangle$,$ \vert c \rangle$. Now, if we keep the unitaries fixed, finding a better choice for the source states can be seen as finding the best rank one approximation with respect to a $4 \times 4 \times 4$ system. After updating the source states, we have to optimize for the unitaries. To compute the optimal choice for $U_{A}$, define the new state $\vert \tilde{\psi} \rangle = \openone_{A} \otimes U_{B} \otimes U_{C} \vert \psi \rangle$. Then one has~\cite{ITN_kraft_2021}
\begin{align}
    \underset{U_{A}}{\text{max}} \, \vert \text{Tr}[U_{A}  \otimes \openone_{B} \otimes \openone_{C} \vert \tilde{\psi} \rangle \langle abc \vert ] \vert = \underset{U_{A}}{\text{max}} \,  \vert \text{Tr}_{A} [U_{A} \rho_{A}] \vert
\end{align}
with $\rho_{A} = \text{Tr}_{BC}[\vert \tilde{\psi} \rangle \langle abc \vert ]$. From the singular value decomposition of $\rho_{A} = UDV^{\dagger}$ one can then derive the optimal form of $U_{A}$ as $U_{A} = VU^{\dagger}$. Alternating this two optimizations multiple times gives a good approximation to the optimal state. However, it should  be noted that similar to the rank one approximation routine, this algorithm is prone to local maxima. In order to stabilize the optimal solution one should use multiple random initial choices. 

\subsection{Modified algorithm and results}
In order to maximize the measure $\mathcal{T}$ we initially choose a random starting point $\vert \psi \rangle$. We use the seesaw algorithm to compute the best network state approximation, called $\vert \omega \rangle$. Then we update the state according to $\vert \psi \rangle \mapsto (1/ \mathcal{N}) (\vert \psi \rangle + \theta \vert \eta \rangle)$ with $\vert \eta \rangle = (\openone - \vert \omega \rangle \langle \omega \vert) \vert \psi \rangle$ and $\mathcal{N}\geq 1$ a normalization. This procedure is iterated.

In the simplified case of only two parties and two independent sources, the problem of triangle preparability reduces to the multilevel entanglement problem~\cite{mutlilevel_entanglement_kraft_2018}, see Fig.~\ref{fig:triangle} (left). In that case, the states which have the smallest overlap among the set of all preparable states are known and can be analytically derived. In the case of four qubits, the state
\begin{align}
    \vert \xi \rangle = \sqrt{\frac{3}{4}} \vert 00 \rangle + \frac{1}{2 \sqrt{3}} (\vert 11 \rangle + \vert 22 \rangle + \vert 33 \rangle)
\end{align}
has the largest distance to the set of decomposable states. Indeed, this state is found by the algorithm with very high fidelity.

For the ITN scenario, we choose $20$ random starting points and a step size of $\theta = 0.1$. Further, to compute the best network state approximation, i.e, the optimal unitaries and source states, we  choose $100$ iterations in the seesaw routine and $500$ different initial points.  This maximization procedure is iterated $200$ times. Then, the step size is changed to $\theta = 0.01, 0.001,0.0001$ while we also increase the precision of the computation of the approximation. While the seesaw routine always makes $100$ iterations, the number of random starting points was $1000,1500,3000$ respectively. For each of the choices, the maximization was done for $200$ steps. After $800$ optimization steps in total there was no increase in the measure $\mathcal{T}$ and the fluctuations were of the same size as those arising in the computation of the best network state approximation. 
The convergence of the algorithm is illustrated in Fig.~\ref{fig:triangle} (right). The state found by the algorithm, denoted by $\vert \psi_{\text{max}} \rangle$, yields a triangle measure $\mathcal{T}$ very close to $0.6$, see Tab.~\ref{tab:triangle}. This is very interesting as this state outperforms all states that have a high triangle measure known so far~\cite{ITN_kraft_2021}. Further, the state shares an interesting entanglement structure. One finds that while the one-body marginals onto the first five systems have all the same spectrum $(0.4,0.6)$, the marginal onto system six is approximately pure, hence unentangled with the remaining particles.

\begin{table*}[ht]
\centering
\begin{tabular}[t]{lcccccccc}
\hline
\hline
State $\vert \psi \rangle$ & $\text{GHZ}_{2}$ & $\text{GHZ}_{3}$ & $\text{GHZ}_{4}$ & $\text{W}$ & $\text{AME}(3,4)$ & $\vert \psi_{3,4} \rangle $ & $\text{AS}_{3}$ & $\vert \psi \rangle_{\text{max}}$\\
\hline
Measure $\mathcal{T}$ & $\tfrac{1}{2}$ & $\tfrac{5}{9} \approx 0.555$ & $\frac{1}{2}$ & $\frac{1}{3} \approx 0.333$ & $\frac{1}{2}$ & $0$ & $\approx 0.4638$ & $0.6$ \\
\hline
\hline
\end{tabular}
\caption{Results of the seesaw optimization yielding upper bounds on the triangle measure $\mathcal{T}$. The states $\text{AME}(3,4)$ and $\vert \psi_{3,4} \rangle$ are defined in the main text. $\text{AS}_{3}$ is the totally antisymmetric state on three qutrits, embedded into the ququad system. $\vert \psi_{\text{max}} \rangle$ refers to the state found by the algorithm.}
\label{tab:triangle}
\end{table*}

%%%%%%%%%%%%%%%%%%%%%%%%%%%%%%%%%%%%%%%%%%%%%%%%%%%%%%%%%%%%%
%%%%%%%%%%%%%%%%%%%%%%%%%%%%%%%%%%
%%%%%%%%%%%%%%%%%%%%%%%%%%%%%%%%%%%%%%%%%%%%%%%%%%%%%%%%%%%%%

%%%%%%%%%%%%%%%%%%%%%%%%%%%%%%%%%%%%%%%%%%%%%%%%%%%%%%%%%%%%%
%%%%%%%%%%%%%%%%%%%%%%%%%%%%%%%%%%
%%%%%%%%%%%%%%%%%%%%%%%%%%%%%%%%%%%%%%%%%%%%%%%%%%%%%%%%%%%%%

\section{Performance of the algorithm for the geometric measure}
\label{app:performance}
In the following we explain our choice of the step-size $\theta$, the number 
of iterations and the duration of the computation. First, it should be noticed 
that the algorithm contains the computation of the closest product state as a
subroutine. Because the used seesaw is prone to local maxima, we randomize the 
algorithm, i.e., we run the iteration for many different initial states. The 
number of iterations as well as the number of initial states depends on the
number of parties and the local dimension. The iteration typically converges 
fast, e.g. for three qubits $10$ iterations are sufficient and for five ququads 
$30$ iterations. The number of initial states can be chosen small if the system size is small, e.g., for three qubits $10$ different initial points make the 
largest overlap robust, while for larger systems more initial states are 
necessary, e.g., for five ququads $100$ points have to be taken. 

The step-size $\theta$ used in the update rule Eq.~\eqref{app:nice_update} 
depends on the size of the system and on the variation of the measure of the 
iterates. For systems of small and moderate size, we initially choose 
$\theta=0.01$. After a certain number of iterations (mostly around $400$) 
the measure of the iterates is not increasing anymore, but fluctuates around 
a certain value where the amount of fluctuation depends on the step-size. 
In this case, the step-size is reduced via $\theta = \theta/2$ and one 
proceeds with the new step-size. However, if $\theta$ becomes small 
it is also useful to improve the precision in the computation of the best
product state approximation. 

\begin{figure}[b]
    \centering
    \includegraphics[width=0.9 \columnwidth]{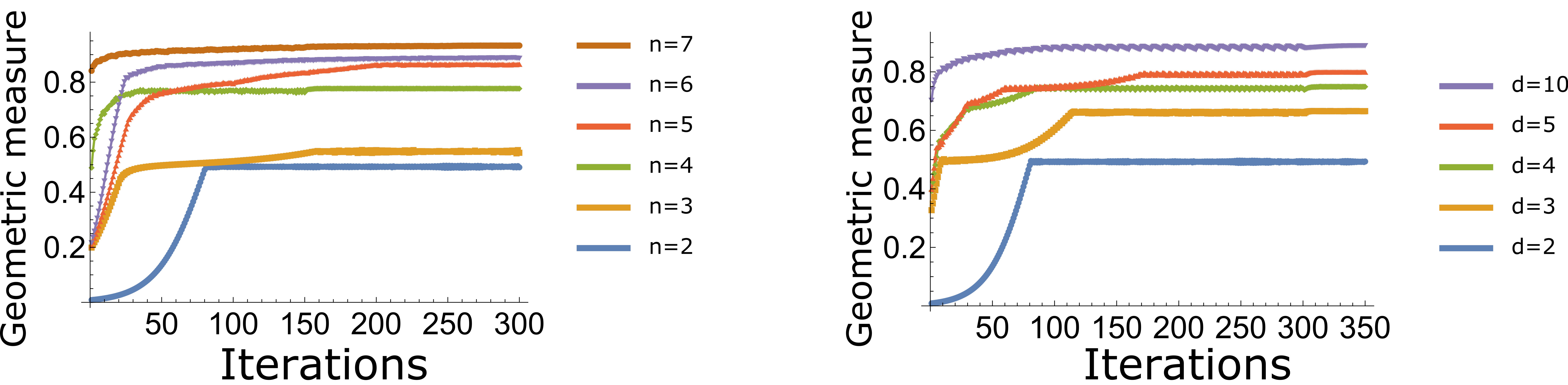}
    \caption{Performance of the algorithm for multi-qubit systems (left). Initializing with a random state, for each iteration the geometric measure $G$ is computed. The step-size $\theta$ depends on the size of the system and we have
    $\theta_{n=2}=\theta_{n=3} = \theta_{n=4} =  0.01$ and $\theta_{n=5} = \theta_{n=6} = 0.06$. For the 
    case $n=3$, the slope starts to decrease when a measure of $G \approx 0.5$ 
    is reached, resulting from the fact that the GHZ state is an exceptional point 
    of the function and yields $G(\vert \text{GHZ} \rangle) =0.5$. A similar 
    behaviour can be observed in the case $n=4$ for the $\vert \text{M} \rangle $ and 
    $\vert \text{L} \rangle $ state. Convergence of the algorithm for bipartite systems of different local dimension $d \in \lbrace 2,3,4,5,10 \rbrace$ (right). For local dimension $d \leq 5$ we have chosen the step-size  as $\theta=0.01$. For $d>5$ the step-size was chosen as $\theta=0.1$. After $350$ iterations, the iterates had a very high fidelity with the $d$-dimensional maximally entangled state.}
    \label{fig:performance_algorithm}
\end{figure}

%\begin{figure}[!htb]
%    \centering
%    \begin{minipage}{.4\textwidth}
%        \centering
%       \includegraphics[width=1\linewidth]{Figure1.png}
%        \caption{Performance of the algorithm for multi-qubit systems. Initializing with a random state, for each iteration the geometric measure $G$ is computed. The step 
%    size $\theta$ depends on the size of the system and we have
 %   $\theta_{n=2}=\theta_{n=3} = 0.01$ and $\theta_{n=6} = 0.06$. For the 
 %   case $n=3$, the slope starts to decrease when a measure of $G \approx 0.5$ 
 %   is reached, coming from the fact that the GHZ state is an exceptional point 
 %   of the function and yields $G(\vert \text{GHZ} \rangle) =0.5$. A similar 
 %   behaviour can be observed for $n=4$ for the $\vert \text{M} \rangle $ and 
%    $\vert \text{L} \rangle $ state.}
%        \label{fig:prob1_6_2}
%    \end{minipage}%
%     \begin{minipage}{.05\textwidth}
%     \mbox{ }
%      \end{minipage}
%    \begin{minipage}{0.4\textwidth}
%        \centering
%        \includegraphics[width=1\linewidth]{Figure2.png}
%        \caption{Convergence of the algorithm for bipartite systems of different local dimension $d \in \lbrace 2,3,4,5,10 \rbrace$. For local dimension $d \leq 5$ we have chosen the step-size  as $\theta=0.01$. For $d=10$ the step-size was chosen as $\theta=0.1$. After $350$ iterations, the iterates had a very high fidelity with the qudit Bell state.}
%        \label{fig:prob1_6_1}
%    \end{minipage}
%\end{figure}

Because the update rule resembles the idea of a gradient descent (GD), we can 
make use of advanced gradient descent techniques as the cumulated momentum method 
(CM)~\cite{gradient_polyak_1964} or Nesterov's accelerated gradient 
(NAG)~\cite{nesterov1983}. In the general setting one aims to minimize 
a (typically complicated and high-dimensional) real-valued function $f(\omega)$. 
Suppose that one initializes the parameters to $\omega_{0}$. In the GD, 
the parameters are iteratively updated according to 
$\omega_{t+1} = \omega_{t} - \kappa_{t}$, where $\kappa_{t} = \theta_{t} (\nabla_{\omega} f) (\omega_t)$ 
with step-size $\theta_{t} >0$, and $(\nabla_{\omega}f)(\omega_{t})$ is the gradient of $f$ 
evaluated at $\omega_{t}$. The idea behind CM is to keep track about the directions we are moving 
in parameter space. For a given momentum parameter $\gamma \in [0,1]$ the update is given 
by $\kappa_{t} = \gamma \kappa_{t-1} + \theta_{t} (\nabla_{\omega}f) (\omega_{t})$. 
Clearly, $\kappa_{t}$ is a running average of previously calculated gradients. 
The advantage of using CM is that the momentum enhances the decrease in directions with small 
gradients while suppressing oscillations in high-curvature directions. 

For NAG, the idea is to not calculate the gradient with respect to the current parameter 
$\omega_{t}$, that is $(\nabla_{\omega}f)(\omega_{t})$, but at the expected value of the 
parameters given the recent momentum, that is, $(\nabla_{\omega}f)(\omega_{t}) + \gamma \kappa_{t-1}$. 
Accordingly, the NAG update rule is given by $\omega_{t+1} = \omega_{t} - \kappa_{t}$ where $\kappa_{t}= \gamma 
\kappa_{t-1} + \theta_{t} (\nabla_{\omega}f)(\omega_{t} + \gamma \kappa_{t-1})$.

The update rule in Eq.~\eqref{app:nice_update} naturally invites for two different variants. 
First, as introduced in the main text, the update direction 
$\vert \eta \rangle = (1/\sqrt{\langle \psi \vert \Pi \vert \psi \rangle}) \Pi \vert \psi \rangle$ 
can be re-normalized. This has the consequence that the amount of the shift stays 
constant for all iterations. However, a second option would be to take 
just the projected (un-normalized) state $\Pi \vert \psi \rangle$ as the update. 
Here, the size of the shift changes within each iteration, depending on how large 
the state $\vert \psi \rangle$ is supported within the subspace $\text{im}(\Pi)$.

%%%%%%%%%%%%%%%%%%%%%%%%%%%%%%%%%%%%%%%%%%%%%%%%%%%%%%%%%%%%%
%%%%%%%%%%%%%%%%%%%%%%%%%%%%%%%%%%
%%%%%%%%%%%%%%%%%%%%%%%%%%%%%%%%%%%%%%%%%%%%%%%%%%%%%%%%%%%%%

\section{Unitary optimization} \label{app:uniopt}
After a sufficient number of iterations, the algorithm yields a state given by 
coordinates with respect to a random basis. Hence, generically each component 
of the tensor is nonzero. 
However, in order to understand the structure of the state, we seek for a concise 
representation in which most of the coefficients vanish. As we consider two states 
to be equal if there is a LU-transformation connecting them, this requires a parametrization 
of the set of unitary matrices. First notice, that $\text{U}(d)$ is the semidirect product 
of $\text{U}(1)$ with $\SU(d)$ and hence we can restrict to parametrizations of $\SU(d)$. For 
qubits, one can make use of the fact that $\SU(2)$ is diffeomorphic to the $3$-sphere 
$S^{3}$. In particular, an arbitrary $\SU(2)$ matrix can be written as 
\begin{align}
    U = 
    \begin{pmatrix}
    \alpha & - \beta^{*} \\
    \beta & \alpha^{*}
    \end{pmatrix},
    \quad  \alpha, \beta \in \mathbb{C} \, \, \text{with} \, \, \vert \alpha \vert^{2} + \vert \beta \vert^{2} =1.
\end{align}
Consequently, the parametrization involves four real parameters and one quadratic constraint. 
However, for $d \geq 3$ the $\SU(d)$ is not homeomorphic to a sphere anymore, or a product of
them, e.g., $\SU(3)$ is \textit{not} homeomorphic to $S^{8}$. Consequently, for systems of
higher dimensions different approaches exist \cite{tilma2002,jarlskog2005,spengler2010}. 
In our work, we have used the Jarlskog parametrization~\cite{jarlskog2005}, what is a 
simple recursive scheme for parametrization, which can be easily implemented 
numerically. First, notice that any $X \in \text{U}(d)$ can be written as 
$X = \Phi_{\alpha} Y \Phi_{\beta}$ where 
$\Phi_{\alpha} = \text{diag}(e^{i \alpha_{1}},...,e^{i \alpha_{d}})$, $\Phi_{\beta}$ 
similar and $Y$ a unitary $d\times d$ matrix. Now, $Y$ is decomposed into a product of unitaries, that is, 
$Y = \prod_{k=2}^{d} A_{d,k} $ with 
\begin{align}
  &  A_{d,k} = 
    \begin{pmatrix}
    A^{(k)} & 0 \\
    0 & \openone_{d-k}
    \end{pmatrix},
    \nonumber
    \\
    &
    A^{(k)} =
    \begin{pmatrix}
    \openone_{d-1} - (1-\cos(\theta_{k})) \vert a_{k} \rangle \langle a_{k} \vert & \sin(\theta_{k}) \vert a_{k} 
    \rangle \\
    - \sin(\theta_{k}) \langle a_{k} \vert & \cos(\theta_{k})
    \end{pmatrix},
\end{align}
where $ A^{(k)} \in \text{U}(d)$, and  $\vert a_{k} \rangle \in \mathbb{C}^{d-1}$ normalized to one, i.e., 
$\langle a_{k} \vert a_{k} \rangle =1$ and $\theta_{k} \in [0,2\pi)$ an arbitrary angle. 

We now describe how this parametrization can be used to bring the numerically found 
states into a concise form. Here, two different cases can be considered. If one has 
a guess for the possible state, e.g., the marginals are all maximally mixed so one 
expects an AME or $k$-uniform state, one could compute the fidelity between the numerical state  
$\vert \psi \rangle$ and the guess $\vert \varphi_{\text{guess}}\rangle$, 
i.e., $\text{sup} \vert \langle \psi \vert U_{1} \otimes \cdots \otimes U_{n} \vert \varphi_{\text{guess}} \rangle \vert$. 
If there is no possible candidate, the idea is to minimize a function 
$f: \text{U}(d) \times \cdots \times \text{U}(d) \rightarrow \mathbb{R}$ 
depending on the state, which becomes minimal if many entries of the state vanish. 
For instance, given the state $\vert \psi \rangle$ a natural candidate 
would be $f(U_{1},...,U_{n}) = \sum \vert  ( U_{1} \otimes \cdots \otimes U_{n} \vert \psi \rangle)_{i_{1},...,i_{n}} \vert $. Given two states $\vert \phi \rangle, \vert \psi \rangle$ we regard them as equal, if $\mathcal{F} (\vert \psi \rangle, \vert \phi \rangle) \geq 1-\epsilon$ with $\epsilon <10^{-6}$, where $\mathcal{F}$ denotes the fidelity.

%%%%%%%%%%%%%%%%%%%%%%%%%%%%%%%%%%%%%%%%%%%%%%%%%%%%%%%%%%%%%
%%%%%%%%%%%%%%%%%%%%%%%%%%%%%%%%%%
%%%%%%%%%%%%%%%%%%%%%%%%%%%%%%%%%%%%%%%%%%%%%%%%%%%%%%%%%%%%%

\section{Relation to known upper bounds}\label{app:upper_bounds}
In this section we discuss fundamental upper bounds on the geometric measure. More precisely, for a given physical system $\mathcal{H} = \otimes_{k} \mathcal{H}_{k}$ with $\mathcal{H}_{k} \cong \mathbb{C}^{d_{k}}$ one can assign the following value
\begin{align}
    \mathcal{G}(\mathcal{H}) := \sup \lbrace G(\vert \varphi \rangle) \, : \, \vert \varphi \rangle \in \mathcal{H} \rbrace
\end{align}
which is characteristic for the space $\mathcal{H}$. Indeed, it can be shown that $\mathcal{G}(\mathcal{H})$ is directly related to the inradius of $B_{\mathcal{H}_{1}} \hat{\otimes} \cdots \hat{\otimes} B_{\mathcal{H}_{n}}$, where $B_{\mathcal{H}_{j}}$ denotes the unit ball in $\mathcal{H}_{j}$ and $\hat{\otimes}$ is the projective tensor product~\cite{aubrun2017}, which is for two closed convex sets $K_{1},K_{2}$ defined as $K_{1} \hat{\otimes} K_{2} = \text{conv} \lbrace x \otimes y \, \vert \, x \in K_{1}, y \in K_{2} \rbrace$. We now want to find upper (tight) bounds on $\mathcal{G} (\mathcal{H})$, bounding the maximal amount of entanglement that can be present in the system. However, similar to the question of maximal tensor rank, this turns out to be a hard question with only partial answers so far~\cite{bounding_maximal_entanglement_qi_2018}. A trivial upper bound can be obtained directly from the normalization of the state. Indeed, given an arbitrary state $\vert \psi \rangle \in (\mathbb{C}^{d})^{\otimes n}$ there must exist a tensor product of computational basis states $\vert j_{1} \cdots j_{n} \rangle$ such that $\vert \langle \psi \vert j_{1} \cdots j_{n} \rangle \vert^{2} \geq d^{-n}$. Consequently we have that $\mathcal{G}( (\mathbb{C}^{d})^{\otimes n}) \leq 1 - d^{-n}$. Although the derivation of this bound is trivial, it is difficult to obtain improved bounds. For instance, one could also consider upper bounds from a convex relaxation
\begin{align}
  &\underset{\vert \psi \rangle}{\text{inf}}  \, \, \underset{\vert \pi \rangle}{\text{sup}} \,  \vert \langle \psi \vert \pi \rangle \vert^{2} = \underset{\rho \in \mathcal{S}}{\text{inf}}  \, \, \underset{\sigma \in \text{SEP}}{\text{sup}} \,  \text{Tr}[\rho \sigma] 
  \nonumber \\ 
  &\geq \underset{\sigma \in \text{SEP}}{\text{sup}}  \underset{\rho \in \mathcal{S}}{\text{inf}} \, \text{Tr}[\rho \sigma] =  \underset{\sigma \in \text{SEP}}{\text{sup}}  \, \lambda_{\text{min}} (\sigma) = \frac{1}{d^{n}},
\end{align}
where $\mathcal{S}$ is the set of mixed states associated to $\mathcal{H}$ and $\text{SEP} \subset \mathcal{S}$ the set of separable states. The first equality is true due the convexity of the objective function and the fact, that pure states / product states are the extreme points of $\mathcal{S}$ / SEP. Further, for any bounded function $f(x,y)$ one has $\sup_{x} f(x,y) \geq f(x,y) \Rightarrow \inf_{y} \sup_{x} f(x,y) \geq \inf_{y} f(x,y)$ and thus we obtain $\inf_{y} \sup_{x} f(x,y) \geq \sup_{x} \inf_{y} f(x,y)$. The last equality is attained for the maximally mixed state. However, it should be noted that this coincides with the trivial upper bound. This bound is not tight, even not for two qubits.  

Another open problem regards the asymptotic scaling of the geometric measure of multi-qubit systems (see Problem 8.27 in \cite{aubrun2017}). Here the question is whether there exists a constant $C >0$ and for any $n \geq 1$ a quantum state $\vert \psi \rangle \in (\mathbb{C}^{2})^{\otimes n}$ such that $G(\vert \psi \rangle) \geq 1 - 2^{-n} C$. If $\vert \psi \rangle \in (\mathbb{C}^{2})^{\otimes n}$ is randomly chosen according to the Haar measure, with high probability it fulfills $G(\vert \psi \rangle) \geq 1 - 2^{-n}Cn\log(n)$, i.e., there is a parasitic factor $n \log(n)$~\cite{aubrun2017}. Note that this bound is an improved version of Eq.~\eqref{eq:bound_low_entanglement} if the number of particles is large. Although we cannot solve that problem here, we can put bounds on the constant $C$ under the assumption that the optimal states found by the algorithm are indeed the global optimizers of the geometric measure. If $G_{\text{max},n}$ denotes the maximal measure of a $n$ qubits system, we have $C_{n} \geq C_{\text{min},n} = 2^{n}(1-G_{\text{max},n})$. One can see in Tab.~\ref{tab:lower_bounds}  that the lower bound $C_{\text{min}}$ of the constant $C$ grows with the number of qubits. In particular, the increase of $C_{\text{max}}$ is not decreasing with the number of qubits $n$ as one would expect if there would exist a constant $C$ which is independent of $n$. Indeed, in this case one would have $C_{\text{min}} \leq \Tilde{C} < \infty$ for all $n$ and as optimal constant $C$ one can simply take the maximum (what exists by assumption) over all the $C_{\text{min}}$.

\begin{table}[t]
\centering
\begin{tabular}[t]{lcc} 
\hline
\hline
System & $G_{\text{max}}$ & $C_{\text{min}}$\\
\hline
$2$ qubits & $1/2$ & $2$\\
$3$ qubits & $5/9$ & $32/9 = 3.555$\\
$4$ qubits & $5/9$ & $32/9 = 3.555$\\
$5$ qubits & $(1/36)(33-\sqrt{3})$ & $(8/9)(3+ \sqrt{3}) \approx 4.206$\\
$6$ qubits & $16/3$ & $16/3= 5.333$\\
$7$ qubits & $0.941$ & $\approx 7.552$\\
$8$ qubits & $0.961$ & $\approx 9.984$\\
\hline
\hline
\end{tabular}
\caption{The maximal entanglement $G_{\text{max}}$ that can be present in a certain system according to the algorithm. This can be used to put lower bounds $C_{\text{min}}$ on the constant $C$ in the asymptotic scaling of the geometric measure in multi-qubit systems. For the cases where no analytical expression of the quantum state and thus of $G_{\text{max}}$ could be derived, i.e., $7$ and $8$ qubits, we have taken the value for $G_{\text{max}}$ for the state found by the algorithm with the largest geometric measure. For each case, the algorithm was run for $100$ random initial states while the geometric measure of the final states coincided with very small deviation.}
\label{tab:lower_bounds}
\end{table}%

Also for the special case of symmetric states an upper bound on the geometric measure can be derived~\cite{upper_bound_symm_measure_friedland}. 
Here it should be noted that the space of symmetric quantum states is much smaller then the space of all states. Indeed, while the state space for $n$ qubits has dimension $2^{n}$, the space of symmetric states only has dimension $n+1$. 
In particular, one can show that for a $n$ qudit system, the geometric measure is upper bounded by $G \leq 1 - 1/c$ where $c = \binom{n+d-1}{n}$, which follows from a similar normalization argument as the upper bound for generic states, but restricting to the Dicke basis. For $n$ qubits this yields the bound $G \leq 1 - 1/(n+1)$, which is at least for a small number of parties not tight. However, it is interesting to see that for $n \geq 5$, the maximal entangled state found by the algorithm violates that bound.

The maximal entangled state with respect to the geometric measure in a certain tensor space $\mathcal{H} = \otimes_{k} \mathcal{H}_{k}$ has also an interpretation from the viewpoint of normed vector spaces. First note that $\vert \vert \psi \vert \vert_{\sigma} = \text{sup}_{\vert \pi \rangle} \vert \langle \pi \vert \psi \rangle \vert$ is a norm on $\mathcal{H}$. Further, because $\mathcal{H}$ is finite dimensional, all norms on $\mathcal{H}$ are equivalent, i.e., if $\vert \vert \cdot \vert \vert $ and $\vert \vert \vert \cdot \vert \vert \vert$ are two norms on $\mathcal{H}$, then there exist constants $C_{1}, C_{2} >0$ such that $C_{1} \, \vert \vert \psi \vert \vert \leq \vert \vert \vert \psi \vert \vert \vert \leq C_{2} \, \vert \vert \psi \vert \vert$ for all (possible un-normalized) vectors $\psi \in \mathcal{H}$. In our case, the two norms are given by $\vert \vert \cdot \vert \vert_{\sigma}$ and $\vert \vert \cdot \vert \vert_{2}$. From this perspective, our algorithm is solving the problem
\begin{align}
   C := \text{min} \, \lbrace \frac{\vert \vert \psi \vert \vert_{\sigma}}{\vert \vert \psi \vert \vert_{2}}  \, : \, \psi \in \mathcal{H} \rbrace
\end{align}
Therefore, knowing $C>0$ implies that $\vert \vert \psi \vert \vert_{\sigma} \, \vert \vert \psi \vert \vert_{2}^{-1} \geq C$ for all $\psi$, thus $C\vert \vert \cdot \vert \vert_{2} \leq \vert \vert \psi \vert \vert_{\sigma}$. Because the algorithm yields the minimizing state explicitly the inequality is tight, hence the found constant $C$ is optimal.

\bibliography{bib.bib}

\end{document}